\documentclass[12pt]{article}
\usepackage[margin=1.05in]{geometry}
\usepackage{amsmath,amssymb,amsthm}
\usepackage{booktabs}
\usepackage{graphicx}
\usepackage{microtype}
\usepackage{mathptmx}
\usepackage[superscript,nomove]{cite}
\usepackage[hidelinks]{hyperref}
\usepackage{caption}
\usepackage{algorithm}
\usepackage{algpseudocode}
\usepackage[framemethod=TikZ]{mdframed}
\newtheorem{theorem}{Theorem}
\newtheorem{proposition}{Proposition}
\newtheorem{lemma}{Lemma}
\newtheorem{corollary}{Corollary}
\theoremstyle{remark}
\newtheorem{remark}{Remark}

\newcommand{\E}{\mathbb{E}}
\newcommand{\Var}{\operatorname{Var}}
\newcommand{\prem}{\Pi}

\title{How Wrong Can a Rank-Based Sample-Size Calculation Be?\\ A Sharp Bound of $16/9$ for Ordinal Outcomes}
\author{%
  Akarin Phaibulpanich\thanks{%
    \textbf{Correspondence:} Akarin Phaibulpanich,
    Department of Statistics, Faculty of Commerce and Accountancy,
    Chulalongkorn University, Bangkok, Thailand.
    \textbf{Email:} akarin@cbs.chula.ac.th}%
  \\[4pt]
  \small Department of Statistics, Faculty of Commerce and Accountancy,\\
  \small Chulalongkorn University, Bangkok, Thailand%
}
\date{\today}
\newcommand{\preprintnote}{%
\begin{center}\small
\fbox{\parbox{0.92\textwidth}{\centering
\textbf{Preprint.} This manuscript has been submitted for publication and is
under review at \emph{Statistics in Medicine}. This version has not been peer
reviewed. Comments and corrections are welcome.}}
\end{center}\medskip}

\begin{document}
\maketitle
\preprintnote

\begin{abstract}
Rank-based tests need no distributional assumptions to be valid, but the sample size they require does depend on the shapes of the two outcome distributions, which are unknown at the design stage. Standard practice substitutes a variance calibrated under the null. We ask how wrong that substitution can be. Writing $\prem$ for the ratio of the true asymptotic variance of the estimated relative effect to the substituted one, we prove that under balanced allocation $\prem \le 16\theta(1-\theta)/\{2+\theta(1-\theta)\} \le 16/9$ for all ordinal distributions in any number of categories: within the first-order asymptotic calculation the required sample size can exceed the calculated one by at most $77.8\%$, and an explicit two-point family attains the bound at every effect size. The effect-specific envelope, not the constant, is the operative quantity for a design: the ceiling falls to $1.63$ at $\theta=0.65$ and $1.37$ at $\theta=0.75$, so a trial powered for a larger effect is correspondingly less exposed. The proof is elementary. In five extracted trial comparisons and $4{,}711$ generated alternatives the realized cost stays within a few percent, because the extremal configuration is one that shift-type treatment mechanisms do not produce. These facts are complementary, and fix how the bound should be used: it is a design sensitivity certificate, not an inflation factor. A protocol can report the conventional sample size alongside the largest requirement consistent with any ordinal configuration.
\end{abstract}


\medskip

\section{Introduction}\label{sec:intro}

Consider investigators planning a randomized trial with an ordinal endpoint, the modified Rankin Scale, to be analysed with a Wilcoxon--Mann--Whitney (WMW) type test. They chose the test to avoid distributional assumptions; yet their first design task, the sample-size calculation, immediately demands one, because the variance driving the calculation depends on the shapes of the two outcome distributions and not only on the anticipated treatment effect.

This is an information constraint, not a defect in existing methodology. The exact planning variance depends on the alternative distributions $p$ and $q$, which are precisely what the trial is being run to learn; at the design stage they are unavailable in principle, not merely inconvenient to obtain. Standard practice therefore substitutes a variance calibrated under the null, computed from an anticipated (or, in blinded re-estimation, an observed) pooled distribution. The substitution is well motivated: it is exact when the two arms share a distribution (Lemma~\ref{lem:null}), and, unlike the alternative variance, it is estimable without unblinding. Away from the null the two quantities differ, and the literature is explicit that they do---Happ, Bathke, and Brunner\cite{happ2019} note the discrepancy directly. The substitution is unavoidable in practice, so the natural question is how large the resulting discrepancy can be. This paper studies that question.

\subsection*{An illustration on three categories}\label{sec:toy}

The following example illustrates the problem before notation is introduced. Suppose the outcome has three ordered categories and, in both scenarios below, treatment beats control in the same fraction of randomly paired patients---the quantity a trialist specifies at the design stage, here $5/8$. The standard calculation therefore returns almost the same sample size for both. The sizes actually required differ by about $69\%$.

\emph{(A) A dispersed shift.} Control patients fall in the three categories with probabilities $(\tfrac38,\tfrac38,\tfrac14)$ and treated patients with $(\tfrac14,\tfrac14,\tfrac12)$: treatment moves some patients up the scale, and both arms remain spread out. Planning for $80\%$ power at $\alpha=0.05$ asks for $N=148$; $139$ would have sufficed.

\emph{(B) A polarizing treatment.} Every control patient sits in the middle category, $(0,1,0)$, while treatment sends patients either up or down the scale, $(\tfrac38,0,\tfrac58)$, leaving nobody in the middle. The same calculation asks for $N=140$. The true requirement is $235$, a shortfall of two-fifths of the trial.

A calculation based only on the anticipated relative effect cannot distinguish these designs, since that effect is identical. What differs is the \emph{shape} of the two distributions, which the calculation does not use. In (B) the treated arm's outcome is essentially a coin flip between two extremes, which makes the true variance large, while the pooled data are concentrated in one category, which makes the calculation's variance estimate small. Those two movements go in opposite directions and multiply. How far apart can they be pushed? That is the question of this paper, and it has an exact answer. Section~\ref{sec:toy2} redoes the example in exact arithmetic once the notation is available.

The familiar slogan that rank-based tests ``require no distributional assumptions'' is thus true of validity and false of design. The test controls its error rate without parametric assumptions, in the nonparametric Behrens--Fisher formulation via studentized statistics \cite{brunner2000,neubert2007}. But the sample size required to detect a given relative effect
\[
\theta \;=\; P(X<Y)+\tfrac12 P(X=Y)
\]
depends on features of the two outcome distributions beyond $\theta$ itself. Happ, Bathke, and Brunner \cite{happ2019} made this precise: asymptotic power is governed by two alternative-hypothesis variance components, and pairs of distributions sharing the same $\theta$ can require materially different sample sizes. Classical bounds on the variance of the Mann--Whitney statistic given $\theta$ go back to Birnbaum and Klose \cite{birnbaum1957}.

Planning practice, however, rarely specifies these components. The most widely implemented calculation, due to Noether \cite{noether1987} with tie corrections \cite{zhao2008}, replaces the alternative variance with a null-calibrated tie variance computed from an anticipated (or, in blinded re-estimation \cite{kieser2003}, an observed pooled) category distribution. Existing work bounds the variance of the Mann--Whitney statistic given $\theta$; we bound the \emph{ratio} of that variance to the null-calibrated variance the calculation substitutes for it, and call that ratio the planning premium.

\begin{table}[b]
\centering
\caption{What the closest prior results establish, and what remains. $\sigma_N^2$ is the variance of $\widehat\theta$; $v(m)$ is the tie-corrected null variance computed from the pooled distribution, which is what the standard sample-size calculation uses.}
\label{tab:novelty}
\small
\begin{tabular}{@{}p{3.0cm}p{4.5cm}p{4.9cm}@{}}
\toprule
Source & Establishes & Leaves open\\
\midrule
Birnbaum--Klose \cite{birnbaum1957}, Rustagi \cite{rustagi1961} & sharp bounds on $\sigma_N^2$ given $\theta$; continuous $F,G$ throughout & ties absent; no comparison to the planning variance $v(m)$\\[2pt]
Bamber \cite{bamber1975}, Brunner--Konietschke \cite{brunner2025} & exact variance identity with ties, including $\theta(1-\theta)-\tau/4$ & an identity, not a bound on a planning ratio\\[2pt]
Happ et al.\ \cite{happ2019} & the components $\sigma_1^2,\sigma_2^2$; optimal allocation; notes the null-for-alternative substitution & does not quantify the resulting error\\[2pt]
Sch\"uurhuis et al.\ \cite{schuurhuis2025} & a test built on $\sigma_N^2$ relative to its Birnbaum--Klose cap $\theta(1-\theta)/m$ & denominator is a cap in $\theta$ alone, used at analysis stage, not $v(m)$\\[2pt]
P\"ohlmann et al.\ \cite{poehlmann2024} & notes the null substitution is conservative for large effects; patches it with simulation-fitted variance functions & no bound on the error, so no way to know when a patch is needed\\
\midrule
\textbf{This paper} & \multicolumn{2}{@{}p{9.6cm}@{}}{sharp envelope for $\prem=(\text{alternative variance})/v(m)$, the ratio the standard calculation actually incurs, with an equality characterization within the two-atom class}\\
\bottomrule
\end{tabular}
\end{table}

\subsection*{Relation to existing work}

Two strands of literature bear directly on this paper, and it is important to be precise about what each already supplies.

\emph{Variance bounds given $\theta$.} That the variance of the Mann--Whitney statistic is constrained by $\theta$ alone is classical. Birnbaum and Klose \cite{birnbaum1957} derived sharp bounds and Rustagi \cite{rustagi1961} sharpened them, obtaining $\Var(\widehat\theta)\le\theta(1-\theta)/\min(n_1,n_2)$, still in current use; both assume $F,G$ continuous, so ties play no role, whereas the tie structure is central here. Bamber \cite{bamber1975} gave an unbiased variance estimator valid under ties and Brunner and Konietschke \cite{brunner2025} its rank-based form, with an identity whose last term is our Lemma~\ref{lem:kernel}. Schüürhuis, Konietschke, and Brunner \cite{schuurhuis2025} build a test around a \emph{ratio} of variances, the true variance to its Birnbaum--Klose maximum $\theta(1-\theta)/m$; that is the closest object in the literature to the premium, but its denominator is a cap depending on $\theta$ alone, and it is used at the analysis stage.

\emph{Rank-based sample-size planning.} Noether \cite{noether1987} gave the calculation used in practice, with tie corrections by Zhao et al.\ \cite{zhao2008} and ordered-categorical treatments by Whitehead \cite{whitehead1993}; Kieser and Friede \cite{kieser2003} developed blinded re-estimation. Happ, Bathke, and Brunner \cite{happ2019} supply the components $\sigma_1^2,\sigma_2^2$, characterize the optimal allocation, and note explicitly that Noether's formula substitutes the null variance for the alternative variance, so that a discrepancy is to be expected. Pöhlmann, Brunner, and Konietschke \cite{poehlmann2024} extend planning to rank-based multiple contrast tests, and their treatment of the variance parameters is the closest precedent for the concern of this paper. Their Method~A substitutes the null value of the variance, $(N+1)/(12N)$ under continuity, and they observe that this is conservative and yields overpowered trials when effects are large; their Methods~B and~C replace it by heuristic parabolic functions $\sigma_i^2(\psi_i)=-\kappa(\psi_i-\tfrac12)^2+(N+1)/(12N)$, with $\kappa$ set to $0.55$ or $0.2$ by simulation. The inadequacy of the null substitution is thus recognized and is already being patched empirically. What we have not found is a statement of how large the error can be, which is what tells a planner whether a patch is needed and how much correction could ever be required.

So the ingredients are all present: the components, the awareness of the substitution, and a classical cap on the numerator. What we have not found is the second step---a sharp comparison of the alternative variance against the \emph{tie-corrected null} variance $v(m)$ computed from the pooled distribution, which is the quantity the standard calculation actually uses. That step does not follow from the Birnbaum--Klose cap, because it requires controlling $v(m)$ jointly with $\theta$ and $\tau$ (our $\Psi\le0$), and the extremal configuration it produces, an atom facing a two-point split, has no continuous analogue. The contribution claimed here is the premium $\prem$ as a planning-stage object, its sharp envelope, an attaining family, and an equality characterization within the two-atom class; the classical cap on the numerator is not ours.

\begin{mdframed}[linewidth=0.7pt,linecolor=black!55,backgroundcolor=black!3,
  innertopmargin=7pt,innerbottommargin=7pt,roundcorner=2pt]
\textbf{Main Theorem (informal).} Under balanced allocation, the premium equals $1$ when the two arms share a non-degenerate distribution; it never exceeds $16\theta(1-\theta)/\{2+\theta(1-\theta)\}$, hence never exceeds $16/9$; and it attains that ceiling at a point mass faced by a two-point comparator straddling it, the only equality configurations among those with at most two atoms per arm.
\end{mdframed}

\medskip
\noindent\emph{Why a finite bound is plausible.} A large premium needs two things at once: a large alternative variance and a small pooled planning variance. These pull against each other. The numerator is large only when one arm's placement scores are spread out, which requires a dispersed comparator; the denominator is small only when the \emph{pooled} distribution is concentrated. Piling both arms onto one category shrinks the numerator to zero; spreading them out inflates the denominator. One configuration balances the two: one arm supplies concentration (a point mass), the other spread (a two-point split), and because the atoms sit on opposite sides of the point mass they generate maximal score variance while adding only two thin categories to the pooled distribution.

\noindent\emph{Interpreting the constant.} Along that family the numerator is exactly $s=\theta(1-\theta)$ and the pooled tie variance exactly $v(m)=(2+s)/32$, so the premium is $16s/(2+s)$, increasing in $s$. Hoeffding's bound caps $s$ at $\tfrac14$, attained at $\theta=\tfrac12$, giving $16/9$. The $2$ in the denominator is the mass the degenerate arm contributes to the pooled distribution, which cannot be dispersed without destroying the numerator.

\medskip

We give a three-part answer for ordinal outcomes with $K$ categories under balanced allocation.

\begin{enumerate}
\item \textbf{The substitution is exact at the null and sharply bounded away from it.} We define the premium $\prem$ (Section~\ref{sec:setup}), show $\prem=1$ whenever the two arms share a non-degenerate distribution (Lemma~\ref{lem:null}), and prove the pointwise envelope of the Main Theorem for every pair of ordinal distributions in any number of categories (Theorem~\ref{thm:global}), attained at every effect size by an explicit two-point family (Theorem~\ref{thm:main}). The proof reduces the infinite-dimensional problem, via a Hoeffding-projection bound, a convexity argument, and a symmetry group, to three two-variable inequalities with hand-checkable decompositions. Within the two-atom class, equality holds only when one arm is a point mass and the other straddles it (Theorem~\ref{thm:equality}), so the extremal configuration is structurally identifiable. Table~\ref{tab:envelope} reports the ceiling at each effect size in both conventional forms, the excess over the calculated size $\prem-1$ and the variance understatement $1-1/\prem$; we quote the former throughout.

\item \textbf{Blinded pooling is conservative for the tie variance itself, with an exact distortion identity.} Pooling two shifted ordinal distributions always makes the blinded data look less tied than the average arm (Proposition~\ref{prop:pool}); for balanced allocation the distortion is exactly $\tfrac1{32}\sum_k (p_k-q_k)^2(p_k+q_k)$, second order in the treatment effect.
\item \textbf{In the settings we could examine, all three prices are small.} For five extracted trial comparisons, and across generated alternatives spanning shift-type and non-shift mechanisms, the realized premium stays within a few percent of one and is usually conservative; for IST-3 as randomised it is $0.9997$ (Sections~\ref{sec:empirical}--\ref{sec:nonshift}). The allocation loss\cite{happ2019} never exceeds $6.4\%$ and the pooled-tie distortion $4.7\%$, and the asymptotic premium predicts realized power to within $2.1$ percentage points from $N=100$ upward (Section~\ref{sec:finite}). Section~\ref{sec:limits} states the scope of this evidence.
\end{enumerate}

The gap between the worst case and the realized values is itself informative. Marginal constraints---caps on the concentration of either arm, even unimodality of both---close only part of it (the Supporting Information). What closes it is the coupling between arms induced by the treatment mechanism, since the extremal configuration requires one arm concentrated at a single category facing an arm polarized on both sides of it. Distribution-free validity does not imply distribution-free planning, but for ordinal endpoints the price of not knowing the shape is capped, and in the settings examined here it is a few percent.

Section~\ref{sec:setup} fixes notation and defines the premium. Section~\ref{sec:theory} contains the theoretical results with proofs. The Supporting Information reports the global numerical search and the effect of interpretable constraints. Section~\ref{sec:empirical} reports the clinical-scenario panel, including allocation and blinded re-estimation diagnostics. Section~\ref{sec:disc} discusses implications for trial design and for the teaching of nonparametric methods, and states limitations.

\section{Setup}\label{sec:setup}

\subsection{Relative effect and variance components}

Let $X\sim p$ and $Y\sim q$ be independent ordinal outcomes on categories $1,\dots,K$, with $p,q$ in the simplex $\Delta_{K-1}$. Define the mid-probability scores
\[
a(x)\;=\;P(Y>x)+\tfrac12 P(Y=x),
\qquad
b(y)\;=\;P(X<y)+\tfrac12 P(X=y),
\]
so that the relative effect satisfies $\theta=\E\,a(X)=\E\,b(Y)=P(X<Y)+\tfrac12P(X=Y)$, with $\theta=\tfrac12$ under exchangeability. The alternative-hypothesis variance components are
\[
\sigma_1^2=\Var\{a(X)\},\qquad \sigma_2^2=\Var\{b(Y)\}.
\]
For allocation fraction $\lambda=n_1/N$ with $n_1$ observations from $p$, the estimated relative effect $\widehat\theta$ satisfies $N\Var(\widehat\theta)\to V_{\mathrm{true}}(\lambda)$ with
\begin{equation}\label{eq:vtrue}
V_{\mathrm{true}}(\lambda)\;=\;\frac{\sigma_1^2}{\lambda}+\frac{\sigma_2^2}{1-\lambda},
\end{equation}
so that the total sample size required for power $1-\beta$ at two-sided level $\alpha$ is, to first order, $N=(z_{1-\alpha/2}+z_{1-\beta})^2\,V_{\mathrm{true}}(\lambda)/(\theta-\tfrac12)^2$; see \cite{happ2019} for the general framework and \cite{brunner2000,neubert2007} for the corresponding studentized inference.

\subsection{Null-calibrated planning and the premium}

For an ordinal distribution $s$ define the tie-corrected null variance
\[
v(s)\;=\;\frac{1-\sum_{k=1}^K s_k^3}{12},
\]
the variance of the mid-distribution transform under a common distribution $s$ (Lemma~\ref{lem:null}); $v(s)=\tfrac1{12}$ in the untied continuous limit. Noether-type planning with ties \cite{noether1987,zhao2008} replaces $V_{\mathrm{true}}$ by the null-calibrated
\begin{equation}\label{eq:vplan}
V_{\mathrm{plan}}(\lambda)\;=\;\frac{v(m)}{\lambda(1-\lambda)},
\qquad
m=\lambda p+(1-\lambda)q,
\end{equation}
where in prospective planning $m$ is an anticipated pooled distribution and in blinded re-estimation it is the observed pooled distribution, which is estimable without unblinding \cite{kieser2003}. The object of study is the \emph{planning premium}
\begin{equation}\label{eq:prem}
\prem(\lambda)
\;=\;
\frac{V_{\mathrm{true}}(\lambda)}{V_{\mathrm{plan}}(\lambda)}
\;=\;
\frac{(1-\lambda)\,\sigma_1^2+\lambda\,\sigma_2^2}{v(m)}.
\end{equation}
at fixed effect, and within the first-order normal approximation, $\prem$ is the multiplicative sample-size correction: the factor by which the null-calibrated calculation under-sizes ($\prem>1$) or over-sizes ($\prem<1$) the trial: a premium of $\prem$ corresponds to $100(\prem-1)\%$ additional patients required, equivalently a variance understatement of $100(1-1/\prem)\%$. Unless stated otherwise we take balanced allocation $\lambda=\tfrac12$, for which $\prem=(\sigma_1^2+\sigma_2^2)/\{2v(m)\}$ and $m=(p+q)/2$.

Two distinct questions must be kept separate. The premium \eqref{eq:prem} isolates the cost of substituting the null-calibrated variance for the true alternative variance; it is not the local-approximation error of the Noether formula at large effects, which is a property of the $(\theta-\tfrac12)^{-2}$ term and is present even when the variance is correctly specified. All statements below concern the former.

\subsection{The three-category example in exact arithmetic}\label{sec:toy2}

The example of Section~\ref{sec:toy} can now be stated exactly; both configurations have $\theta=\tfrac58$. In (A), $p=(\tfrac38,\tfrac38,\tfrac14)$ and $q=(\tfrac14,\tfrac14,\tfrac12)$ give scores $a=(\tfrac78,\tfrac58,\tfrac14)$ and $b=(\tfrac3{16},\tfrac9{16},\tfrac78)$, hence $\sigma_1^2=\tfrac{15}{256}$, $\sigma_2^2=\tfrac{41}{512}$, $m=(\tfrac5{16},\tfrac5{16},\tfrac38)$, $v(m)=\tfrac{605}{8192}$ and $\prem=\tfrac{568}{605}=0.939$. In (B), $p=(0,1,0)$ and $q=(\tfrac38,0,\tfrac58)$ make $a\equiv\tfrac58$ constant, so $\sigma_1^2=0$, while $b=(0,\tfrac12,1)$ gives $\sigma_2^2=\tfrac{15}{64}$; now $m=(\tfrac3{16},\tfrac12,\tfrac5{16})$, $v(m)=\tfrac{143}{2048}$ and $\prem=\tfrac{240}{143}=1.678$, which attains the envelope of Theorem~\ref{thm:global} at $\theta=\tfrac58$. Moving from (A) to (B) roughly triples $\sigma_2^2$, since $b$ becomes an indicator rather than a spread-out score, while shrinking $v(m)$, since half the pooled mass sits in one category.

\section{Theory}\label{sec:theory}

\subsection{Behaviour under the null}

\begin{lemma}\label{lem:null}
If $p=q=s$ then $\sigma_1^2=\sigma_2^2=v(s)$. If moreover $v(s)>0$---that is, $s$ is not a one-point distribution---then $\prem(\lambda)=1$ for every $\lambda\in(0,1)$.
\end{lemma}

\begin{remark}\label{rem:degen}
The proviso is not vacuous: if $p=q=e_j$ then $\sigma_1^2=\sigma_2^2=v(s)=0$ and $\prem=0/0$ is undefined. We adopt the convention $\prem=1$ there by continuous extension, which is the value approached along any path with $p=q$. Note also that Lemma~\ref{lem:null} is a one-way statement: we do not claim, and do not use, the converse that $\prem=1$ implies $p=q$.
\end{remark}

\begin{proof}
Under $p=q=s$, $b(y)=P(X<y)+\tfrac12P(X=y)$ is the mid-distribution transform of $s$ evaluated at $y\sim s$, and $a=1-b$, so it suffices to compute $\Var\{b(Y)\}$. Write $c_x=P(X<x)$. Then $\E\,b(Y)=\sum_x s_x(c_x+\tfrac{s_x}{2})=\tfrac12$ by symmetry, and
\[
\E\,b(Y)^2=\sum_x s_x c_x^2+\sum_x s_x^2 c_x+\tfrac14\sum_x s_x^3 .
\]
Since $(c_x+s_x)^3-c_x^3=3c_x^2s_x+3c_xs_x^2+s_x^3$ telescopes to $1$ over $x$, we get $\sum_x s_xc_x^2+\sum_x s_x^2c_x=\tfrac13(1-\sum_x s_x^3)$, whence $\Var\{b(Y)\}=\tfrac13(1-\sum s_x^3)+\tfrac14\sum s_x^3-\tfrac14=(1-\sum_x s_x^3)/12=v(s)$.
\end{proof}

The null-calibrated planning variance is therefore not an ad hoc approximation: it is the \emph{exact} variance at the boundary of the hypothesis, and $\prem$ measures how far the alternative drifts from it.

\subsection{Bounding the numerator}

The worked example of Section~\ref{sec:toy} suggests that a large planning premium requires two opposing features at once: an alternative variance that is large and a pooled planning variance that is small. This is the strategy of the argument that follows. Proposition~\ref{prop:hoeffding} shows that the first quantity is capped universally, by $\theta$ alone; the remainder of the argument concerns how far the second can be driven down without destroying the first. Theorem~\ref{thm:global} shows the two cannot be separated by more than a factor $16/9$.

\begin{proposition}\label{prop:hoeffding}
For all $p,q$: $\;\sigma_1^2+\sigma_2^2\;\le\;\theta(1-\theta)\;\le\;\tfrac14$.
\end{proposition}

\begin{proof}
Let $h(X,Y)=\mathbf 1\{X<Y\}+\tfrac12\mathbf 1\{X=Y\}\in[0,1]$, so $\theta=\E h$ and $a(X)=\E[h\mid X]$, $b(Y)=\E[h\mid Y]$. The Hoeffding decomposition \cite{hoeffding1948} gives
\[
\Var(h)\;=\;\Var\{a(X)\}+\Var\{b(Y)\}+\E\bigl[\{h-a(X)-b(Y)+\theta\}^2\bigr]\;\ge\;\sigma_1^2+\sigma_2^2,
\]
and since $h\in[0,1]$ implies $h^2\le h$, $\Var(h)=\E h^2-\theta^2\le\theta-\theta^2$.
\end{proof}

This bound is classical, not ours; see Section~\ref{sec:intro} and \cite{birnbaum1957,brunner2025}. We record it in this form only because it is the shape we need.

Proposition~\ref{prop:hoeffding} already explains the qualitative finding of this paper. The numerator of $\prem(\tfrac12)$ is capped by $\theta(1-\theta)/2\le\tfrac18$ regardless of shape, so the premium can be large only where the pooled tie variance $v(m)$ is simultaneously forced to be small---and, as the next results show, the two requirements are nearly incompatible. Related sharp bounds for the continuous Mann--Whitney variance as a function of $\theta$ are classical \cite{birnbaum1957}; Proposition~\ref{prop:hoeffding} is the elementary form suited to our ratio.

\subsection{Behaviour of the planning variance under pooling}

\begin{proposition}\label{prop:pool}
For any $\lambda\in(0,1)$, $v(\lambda p+(1-\lambda)q)\ge \lambda v(p)+(1-\lambda)v(q)$, with equality iff $p=q$. For $\lambda=\tfrac12$,
\begin{equation}\label{eq:pool}
v\!\Bigl(\frac{p+q}{2}\Bigr)-\frac{v(p)+v(q)}{2}
\;=\;
\frac{1}{32}\sum_{k=1}^K (p_k-q_k)^2\,(p_k+q_k)\;\ge\;0 .
\end{equation}
\end{proposition}

\begin{proof}
The first claim is convexity of $t\mapsto t^3$ on $[0,1]$ applied coordinatewise inside $v$. For \eqref{eq:pool}, the scalar identity
$\tfrac{a^3+b^3}{2}-\bigl(\tfrac{a+b}{2}\bigr)^3=\tfrac38(a-b)^2(a+b)$
(direct expansion) applied to $(a,b)=(p_k,q_k)$, summed over $k$ and divided by $12$, gives the result.
\end{proof}

Two consequences. A blinded interim estimate of the pooled tie pattern is biased toward \emph{less} concentration than the average arm, so blinded tie-corrected re-estimation overstates $v$ and over-sizes rather than under-sizes; and the distortion is second order in $p-q$, i.e.\ $O(\|p-q\|^2)$, vanishing at the null. Converting that rate to $O\{(\theta-\tfrac12)^2\}$ requires a path along which $\theta-\tfrac12$ has non-vanishing first derivative. Note what Proposition~\ref{prop:pool} does \emph{not} say: conservatism for $v$ does not imply conservatism for the alternative components $\sigma_1^2,\sigma_2^2$ that govern power. Quantifying that gap is the role of the premium.

\subsection{The extremal family}

\begin{theorem}\label{thm:main}
Let $K\ge 3$ and $\lambda=\tfrac12$.
\begin{enumerate}
\item[(i)] \textup{(Extremal family.)} For $1\le j\le K-2$ and $t\in(0,1)$, the pair
\[
p=e_j,\qquad q=(1-t)\,e_{j-1}+t\,e_{j+1}
\]
\textup{(}$e_j$ the point mass at category $j$\textup{)} has $\theta=t$, $\sigma_1^2=0$, $\sigma_2^2=t(1-t)$, and, writing $s=\theta(1-\theta)$,
\begin{equation}\label{eq:family}
\prem\;=\;\frac{16\,s}{2+s}\;=\;\frac{16\,\theta(1-\theta)}{2+\theta(1-\theta)} .
\end{equation}
\item[(ii)] \textup{(Sharp bound in the degenerate-arm class.)} Over all pairs $(p,q)\in\Delta_{K-1}^2$ in which at least one arm is a point mass,
\[
\sup\,\prem\;=\;\frac{16}{9},
\]
attained exactly, and only, by the family in \textup{(i)} with $t=\tfrac12$ \textup{(}up to position and the arm-reversal symmetry\textup{)}.
\end{enumerate}
\end{theorem}

\begin{proof}
(i) On the support of $p$, $a(j)=P(Y>j)+\tfrac12P(Y=j)=t$, so $a(X)$ is constant and $\sigma_1^2=0$. The score $b$ takes the value $0$ at $j-1$ and $1$ at $j+1$, so $b(Y)\sim\mathrm{Bernoulli}(t)$, giving $\theta=t$ and $\sigma_2^2=t(1-t)=s$. The pooled distribution places masses $\tfrac{1-t}{2},\tfrac12,\tfrac{t}{2}$ on the three categories, so $\sum_k m_k^3=\{(1-t)^3+1+t^3\}/8=(2-3s)/8$ and $v(m)=\{1-(2-3s)/8\}/12=(2+s)/32$. Hence $\prem=\sigma_2^2/\{2v(m)\}=16s/(2+s)$.

(ii) The proof is a self-contained two-variable computation, independent of the machinery of Section~\ref{sec:global}; it is given in the Supporting Information, and part~(ii) is in any case subsumed by Theorem~\ref{thm:global}.
\end{proof}

\subsection{A global upper bound}\label{sec:global}

Theorem~\ref{thm:main} shows that the envelope $16\theta(1-\theta)/\{2+\theta(1-\theta)\}$ is attained; the main result of this section is that it is never exceeded, by any pair of ordinal distributions, in any dimension.

\begin{theorem}[Pointwise envelope]\label{thm:global}
Let $K\ge2$, $\lambda=\tfrac12$, and $(p,q)\in\Delta_{K-1}^2$ with $v(m)>0$. Then, with $s=\theta(1-\theta)$,
\begin{equation}\label{eq:envelope}
\prem\;\le\;\frac{16\,s}{2+s}\;\le\;\frac{16}{9}.
\end{equation}
Consequently, for every $K\ge3$: $\sup_{(p,q)}\prem=16/9$, and for each $\theta_0\in[\tfrac12,1)$---the restriction matters, since $\bar\prem$ is symmetric about $\tfrac12$ and decreasing only on this half---the supremum of $\prem$ subject to $\theta\ge\theta_0$ equals $16\theta_0(1-\theta_0)/\{2+\theta_0(1-\theta_0)\}$; both suprema are attained by the family of Theorem~\ref{thm:main}\textup{(i)}.
\end{theorem}

The proof reduces an infinite-dimensional optimization to six explicit algebra problems in four steps. \emph{Step 1} bounds the numerator by a quantity depending only on $\theta$ and the tie probability (Lemma~\ref{lem:kernel}), after which no variance component appears. \emph{Step 2} shows the worst case is attained with at most two atoms per arm (Lemma~\ref{lem:reduce}), making the problem finite for every $K$ at once. \emph{Step 3} collapses the thirteen resulting arrangements to six by symmetry (Lemma~\ref{lem:sym}). \emph{Step 4} verifies the six: three are immediate, three reduce to the identities of Appendix~\ref{app:ident}.

\subsubsection*{Step 1: remove the variance components}

\emph{Idea.} Both components are conditional expectations of one bounded kernel, so Hoeffding's decomposition bounds their sum by that kernel's variance. The components $\sigma_1^2,\sigma_2^2$ are the obstruction, being functionals of $p$ and $q$ that resist direct optimization; after this step the problem depends on the pair only through $\theta$, $\tau$ and $\sum_kc_k^3$.

\begin{lemma}[Kernel variance; \cite{bamber1975,brunner2025}]\label{lem:kernel}
Let $h=\mathbf 1\{X<Y\}+\tfrac12\mathbf 1\{X=Y\}$ and $\tau=P(X=Y)=\sum_k p_kq_k$. Then $\Var(h)=\theta-\tfrac{\tau}{4}-\theta^2$, and hence, by the Hoeffding decomposition as in Proposition~\ref{prop:hoeffding},
\[
\sigma_1^2+\sigma_2^2\;\le\;s-\frac{\tau}{4}.
\]
\end{lemma}

\begin{proof}
$h^2$ equals $1$ on $\{X<Y\}$, $\tfrac14$ on $\{X=Y\}$ and $0$ otherwise, so $\E h^2=W+\tfrac{\tau}{4}$ with $W=P(X<Y)=\theta-\tfrac{\tau}{2}$; thus $\Var(h)=\theta-\tfrac{\tau}{4}-\theta^2$.
\end{proof}

\subsubsection*{Step 2: reduce infinitely many distributions to finitely many}

Step 1 reduces the problem to a single polynomial inequality, and since everything that follows rests on that polynomial we record its derivation. With $m=(p+q)/2$ and $c=p+q$ we have $\sum_k m_k^3=\tfrac18\sum_kc_k^3$, so
\[
2v(m)\;=\;\frac{1-\tfrac18\sum_kc_k^3}{6}\;=\;\frac{8-\sum_kc_k^3}{48},
\]
and Lemma~\ref{lem:kernel} gives $\prem=(\sigma_1^2+\sigma_2^2)/\{2v(m)\}\le 48(s-\tau/4)/(8-\sum_kc_k^3)$. The envelope $\prem\le16s/(2+s)$ therefore follows from
\[
48\Bigl(s-\frac{\tau}{4}\Bigr)(2+s)\;\le\;16\,s\Bigl(8-\sum_kc_k^3\Bigr),
\]
that is, after dividing by $16$, from $\Psi\le0$ with $\Psi$ as in \eqref{eq:psi} below. Note $8-\sum_kc_k^3>0$ whenever $v(m)>0$, so no sign is reversed in clearing the denominator. Everything after this point is the study of that one object.

\medskip
\noindent\emph{Idea.} Maximizing a convex function over a polytope pushes the optimum to extreme points. Holding $\theta$ fixed makes the admissible set a simplex cut by one hyperplane, and $\Psi$ convex on it, since $s$ is frozen, $\tau$ affine, and only the convex $\sum_k c_k^3$ moves. Such extreme points carry at most two atoms, and the argument never mentions $K$.

\begin{lemma}[Extreme points of a singly-cut simplex]\label{lem:extreme}
Let $u\in\mathbb R^K$ and $\gamma\in\mathbb R$, and let $P=\{x\in\Delta_{K-1}:u^{\!\top}x=\gamma\}$ be nonempty. Every extreme point of $P$ has at most two nonzero coordinates.
\end{lemma}

\begin{proof}
Let $x\in P$ have $|\operatorname{supp}(x)|\ge3$ and pick distinct $i,j,k\in\operatorname{supp}(x)$. Consider directions $d\in\mathbb R^K$ supported on $\{i,j,k\}$ subject to $\sum_\ell d_\ell=0$ and $u^{\!\top}d=0$. These are two linear conditions on a three-dimensional space, so the solution set contains some $d\ne0$. Since $x_i,x_j,x_k>0$, there is $\varepsilon>0$ with $x\pm\varepsilon d\ge0$, and both points satisfy the two equality constraints, hence lie in $P$. As $x=\tfrac12\{(x+\varepsilon d)+(x-\varepsilon d)\}$ with $x+\varepsilon d\ne x-\varepsilon d$, $x$ is not extreme.
\end{proof}

\begin{lemma}[Two-atom reduction]\label{lem:reduce}
Write $c_k=p_k+q_k$ and
\begin{equation}\label{eq:psi}
\Psi(p,q)\;=\;3\Bigl(s-\frac{\tau}{4}\Bigr)(2+s)\;-\;s\,\Bigl(8-\sum_{k}c_k^3\Bigr),
\qquad s=\theta(1-\theta),\quad \tau=\sum_k p_kq_k .
\end{equation}
Then $\sup_{\Delta_{K-1}^2}\Psi$ is attained at a pair in which each arm has at most two support points. The statement and its proof hold for every $K\ge2$, with no constant depending on $K$.
\end{lemma}

\begin{proof}
$\Psi$ is a polynomial on the compact set $\Delta_{K-1}^2$, so a maximizer $(p^*,q^*)$ exists; write $\theta^*=\theta(p^*,q^*)$, $s^*=\theta^*(1-\theta^*)$, and let $a^*$ be the score vector of $q^*$, so that $\theta(p,q^*)=\sum_k p_k a^*_k$ is a linear functional of $p$.

\emph{Step (a): $\Psi(\cdot,q^*)$ is convex on the slice.} Put
\[
P_1=\{p\in\Delta_{K-1}:\textstyle\sum_k p_k a^*_k=\theta^*\},
\]
a nonempty compact convex set containing $p^*$. For $p\in P_1$ the value of $\theta$, and hence of $s$, is the constant $s^*$. Expanding \eqref{eq:psi} and grouping,
\begin{equation}\label{eq:psisplit}
\Psi(p,q^*)\;=\;\underbrace{3s^*(2+s^*)-8s^*}_{\text{constant on }P_1}\;-\;\underbrace{\tfrac34(2+s^*)\sum_k p_kq^*_k}_{\text{affine in }p}\;+\;\underbrace{s^*\sum_k (p_k+q^*_k)^3}_{s^*\ \ge\ 0\ \text{times a convex function}} ,
\end{equation}
an identity in $p$ verified by expansion. The first term is constant, the second is affine because $\tau=\sum_k p_kq^*_k$ is linear in $p$, and the third is a nonnegative multiple of $p\mapsto\sum_k(p_k+q^*_k)^3$, which is convex on $\Delta_{K-1}$ as a sum of the convex functions $p_k\mapsto(p_k+q^*_k)^3$ on $p_k\ge0$. A constant plus an affine function plus a nonnegative multiple of a convex function is convex, so $\Psi(\cdot,q^*)$ is convex on $P_1$ in the full (not merely coordinatewise) sense.

\emph{Step (b): move to an extreme point.} A convex function on a nonempty compact convex subset of $\mathbb R^K$ attains its maximum at an extreme point (Bauer's maximum principle). Hence there is an extreme point $p'$ of $P_1$ with $\Psi(p',q^*)\ge\Psi(p^*,q^*)=\sup\Psi$, so equality holds and $(p',q^*)$ is also a maximizer. By Lemma~\ref{lem:extreme} applied with $u=a^*$ and $\gamma=\theta^*$, $p'$ has at most two nonzero coordinates. Note $\theta(p',q^*)=\theta^*$, since $p'\in P_1$.

\emph{Step (c): repeat in the second argument.} Now fix $p'$ and let $b'$ be its score vector, so $\theta(p',q)=\sum_j q_jb'_j$ is linear in $q$. Set $P_2=\{q\in\Delta_{K-1}:\sum_j q_jb'_j=\theta^*\}$, which contains $q^*$. The decomposition \eqref{eq:psisplit} is symmetric in the roles of the two arms: with $p'$ fixed and $\theta$ frozen at $\theta^*$ on $P_2$, $\Psi(p',\cdot)$ is again constant plus affine ($\tau=\sum_j p'_jq_j$ is linear in $q$) plus $s^*$ times the convex $q\mapsto\sum_k(p'_k+q_k)^3$. Steps (a) and (b) therefore apply verbatim, yielding an extreme point $q'$ of $P_2$ with $\Psi(p',q')=\sup\Psi$ and $|\operatorname{supp}(q')|\le2$. Since this step does not alter $p'$, the pair $(p',q')$ has at most two support points in each arm and attains the supremum.

No step refers to $K$: Lemma~\ref{lem:extreme} holds in every dimension, the nonlinear term in \eqref{eq:psisplit} is a sum of convex coordinate functions and hence globally convex, and Bauer's principle applies to every nonempty compact convex set. \end{proof}

\begin{remark}
Two points deserve emphasis, since they are where such reductions usually fail. First, the convexity in Step~(a) is genuine convexity on the slice, not merely convexity in each coordinate: it follows from \eqref{eq:psisplit}, in which the only non-affine term carries the \emph{nonnegative} coefficient $s^*$. Had $s$ varied on the slice, or entered with a sign, the argument would break; freezing $\theta$ is what makes $s$ constant, and this is the reason the reduction is performed on $\theta$-slices rather than on $\Delta_{K-1}^2$ directly. Second, Step~(c) is legitimate only because fixing $\theta=\theta^*$ in the second stage preserves the value obtained in the first; both stages optimize on the same level set of $\theta$, so the two reductions do not undo one another.
\end{remark}

\subsubsection*{Step 3: identify geometrically equivalent configurations}

\emph{Idea.} Relabelling the arms, or reading the scale in reverse, replaces $\theta$ by $1-\theta$ and leaves $s$, $\tau$ and $p+q$ untouched. Many of the thirteen arrangements are therefore one configuration seen twice.

\begin{lemma}[Symmetry]\label{lem:sym}
$\Psi$ is invariant under exchanging the two arms and under reversing the category order. Consequently the thirteen order/coincidence patterns of two atoms per arm fall into six orbits, and it suffices to verify $\Psi\le0$ on one representative of each.
\end{lemma}

\begin{proof}
Exchanging $p$ and $q$ maps $\theta\mapsto1-\theta$, hence fixes $s=\theta(1-\theta)$, and fixes $\tau$ and $c=p+q$; reversing the category order likewise maps $\theta\mapsto1-\theta$ and permutes the coordinates of $c$. Both therefore fix every term of \eqref{eq:psi}. The enumeration of patterns is complete and is counted as follows. Two atoms per arm occupy between two and four distinct categories. With four distinct positions, the pattern is determined by which two of the four ordered slots carry the $p$-atoms, giving $\binom42=6$; with exactly one shared position there are three ordered slots, and the pattern is fixed by which slot is shared ($3$ choices) and which of the remaining two carries the $p$-atom ($2$ choices), giving $6$; with both positions shared there is $1$. In total $6+6+1=13$, listed in Appendix~\ref{app:cert}. The induced action of the group on these thirteen has orbits
$\{PPQQ,QQPP\}$, $\{PQQP,QPPQ\}$, $\{PQPQ,QPQP\}$, $\{P\!\cdot\!PQ\!\cdot\!Q,\,Q\!\cdot\!PQ\!\cdot\!P\}$, $\{PQ\!\cdot\!P\!\cdot\!Q,PQ\!\cdot\!Q\!\cdot\!P,P\!\cdot\!Q\!\cdot\!PQ,Q\!\cdot\!P\!\cdot\!PQ\}$ and $\{PQ\!\cdot\!PQ\}$.
\end{proof}

\subsubsection*{Step 4: verify the remaining six cases}

\begin{proof}[Proof of Theorem~\ref{thm:global}]
Since $2v(m)=\bigl(8-\sum_k c_k^3\bigr)/48$, Lemma~\ref{lem:kernel} gives $\prem\le 48\,(s-\tfrac{\tau}{4})/(8-\sum_k c_k^3)$, so \eqref{eq:envelope} follows if $\Psi\le0$ on $\Delta_{K-1}^2$. By Lemma~\ref{lem:reduce} it suffices to treat pairs with at most two atoms per arm, and such a pair enters $\Psi$ only through the masses $(\alpha,\bar\alpha)$, $(\beta,\bar\beta)$ and the order pattern of the at most four support points; the verification is therefore uniform in $K$. By Lemma~\ref{lem:sym} six representatives suffice.

Each representative polynomial below is simply what $\Psi$ becomes once the atom positions are fixed, written in the masses $\alpha$ and $\beta$.

\emph{Orbits 1--3 (no shared support, $\tau=0$).} For $P{<}P{<}Q{<}Q$, $\theta\equiv1$, so $s=\tau=0$ and $\Psi\equiv0$. For $P{<}Q{<}Q{<}P$, $\Psi=-3\alpha\bar\alpha\beta\bar\beta\le0$. For $P{<}Q{<}P{<}Q$,
\[
\Psi=-3\,\alpha\bar\alpha\beta\bar\beta\,(1-\bar\alpha\beta)\,(1+2\beta-\alpha-\alpha\beta)\;\le\;0,
\]
since $1-\bar\alpha\beta\ge0$ and $1+2\beta-\alpha(1+\beta)\ge1+2\beta-(1+\beta)=\beta\ge0$.

\emph{Orbits 4--6 (a shared position, or both shared).} These are the hard cases, for a statistical rather than algebraic reason. When the arms share no category, $\tau=0$, $\Psi$ loses its tie term, and the survivor factors into separately signed pieces; this is why orbits 1--3 fall immediately. A shared category enters $\sum_kc_k^3$ with the \emph{sum} of the two arms' masses, cubed, so the tie term and the pooled-variance term move together and neither can be bounded alone. That coupling destroys the factorization, and it is exactly the regime in which the premium is largest, since shared mass is what concentrates the pooled distribution.

Each of the three reduces to showing that an explicit polynomial of bidegree at most $(4,4)$ is nonnegative on the unit square, and each is settled by one algebraic identity that rewrites it as a combination of manifestly signed terms: for the doubly-shared orbit by isolating a perfect square in $\alpha-\beta$, and for the two singly-shared orbits by translating the cofactor's unique zero to the origin, after which both take the same form and succumb to the same three elementary bounds. The identities, the polynomials, and the verification are given in Appendix~\ref{app:ident}.

Finally, $x\mapsto16x/(2+x)$ is increasing and $s\le\tfrac14$, giving the constant $16/9$; the envelope decreases in $\theta$ on $[\tfrac12,1)$, so the $\theta$-constrained supremum is $16\theta_0(1-\theta_0)/\{2+\theta_0(1-\theta_0)\}$, attained by Theorem~\ref{thm:main}(i), and $\theta_0=\tfrac12$ gives the global statement.
\end{proof}

\begin{remark}
The proof uses no floating-point computation and no computer algebra beyond expanding the two identities \eqref{eq:Fdecomp} and \eqref{eq:Sform}, which a reader can verify directly.
\end{remark}

\begin{remark}[Scope of Theorem~\ref{thm:global}]\label{rem:scope}
The statement is confined to ordinal outcomes with finitely many categories and to balanced allocation. Unequal allocation is not covered, and the argument does not extend: \eqref{eq:lam} shows $\prem(\lambda)$ diverges as $\lambda\to1$ along the extremal family, and the reduction of Lemma~\ref{lem:reduce} uses a symmetry of $\Psi$ that $\lambda\ne\tfrac12$ destroys. The bound is on an asymptotic variance ratio (Section~\ref{sec:finite}) and covers only the variance substitution (Section~\ref{sec:notcover}).
\end{remark}

\begin{corollary}[Effect of the tie correction]\label{cor:uncorrected}
Let $\lambda=\tfrac12$ and let $\prem_0=V_{\mathrm{true}}(\tfrac12)\big/\{(1/12)/\lambda(1-\lambda)\}$ denote the premium incurred when planning uses the \emph{uncorrected} null variance $1/12$ in place of $v(m)$. Then
\[
\prem_0\;=\;6(\sigma_1^2+\sigma_2^2)\;\le\;6\Bigl(s-\frac{\tau}{4}\Bigr)\;\le\;6s\;\le\;\frac32 ,
\]
and $3/2$ is attained, by the family of Theorem~\textup{\ref{thm:main}(i)} at $\theta=\tfrac12$. Hence
\[
\sup\prem_0=\frac32<\frac{16}{9}=\sup\prem ,
\]
the two suprema differing by the factor $32/27$.
\end{corollary}

\begin{proof}
Immediate from Lemma~\ref{lem:kernel} and $s\le\tfrac14$; the family has $\tau=0$, $\sigma_1^2=0$, $\sigma_2^2=s$, and Hoeffding is tight there.
\end{proof}

\begin{remark}
The comparison is worth stating because it runs against intuition. The tie correction is an improvement: it replaces $1/12$ by the smaller, exactly correct null variance $v(m)$, and so removes the automatic padding that ties would otherwise supply. But that padding was also insurance. Because the correction lowers the planned sample size, it raises the worst-case ratio of true to planned variance, from $3/2$ to $16/9$. A planner using the uncorrected formula is more conservative on average---which is the behaviour reported for Method~A of \cite{poehlmann2024}---and is also better protected in the worst case, at the price of over-sizing in the typical case. Neither convention is uniformly preferable, but the choice should be made knowingly, and the two ceilings quantify it.
\end{remark}

\begin{proposition}[Strictness in the nontrivial interior]\label{prop:strict}
For each of the orbit representatives $2$--$6$, $\Psi<0$ at every $(\alpha,\beta)\in(0,1)^2$. Orbit~$1$ has $\theta\in\{0,1\}$ and is excluded from the nontrivial equality problem. Consequently, among two-atom configurations with $\theta\in(0,1)$, equality $\Psi=0$ can occur only when $\alpha\in\{0,1\}$ or $\beta\in\{0,1\}$, that is, only when one arm is degenerate.
\end{proposition}

\begin{proof}
Orbit~2 has $\Psi=-3\alpha\bar\alpha\beta\bar\beta<0$ on $(0,1)^2$. Orbit~3 has
$\Psi=-3\alpha\bar\alpha\beta\bar\beta(1-\bar\alpha\beta)(1+2\beta-\alpha-\alpha\beta)$, in which $1-\bar\alpha\beta\ge1-\beta>0$ and $1+2\beta-\alpha(1+\beta)\ge\beta>0$ on the interior, so $\Psi<0$ there.

For orbit~6, the decomposition \eqref{eq:Fdecomp} gives $F=7\{\alpha(1-\alpha)+\beta(1-\beta)\}+(\alpha-\beta)^2R$ with $R\ge0$; on $(0,1)^2$ both $\alpha(1-\alpha)$ and $\beta(1-\beta)$ are strictly positive, so $F>0$ and $\Psi=-\tfrac3{16}F<0$. For orbits~4 and~5, the chain established in the proof of Theorem~\ref{thm:global} yields, in the shifted coordinates of \eqref{eq:Sform},
\[
S(x,y)\;\ge\;7(x+y)-14xy\;=\;7\{x(1-y)+y(1-x)\},
\]
which is strictly positive whenever $x,y\in(0,1)$; since the substitutions $(\alpha,\beta)=(x,1-y)$ and $(1-x,1-y)$ map $(0,1)^2$ onto itself, the corresponding $\Psi$ is strictly negative on the open square. The monomial prefactors removed before those substitutions vanish only on the boundary.
\end{proof}

\subsection{Configurations attaining the bound}\label{sec:equality}

For a design safeguard, which configurations attain the worst case matters as much as its size.

\begin{theorem}[Equality]\label{thm:equality}
Let $\lambda=\tfrac12$ and $\theta\notin\{0,1\}$.
\begin{enumerate}
\item[(i)] \textup{(Sufficiency, general.)} If $p=e_j$ and $q$ is supported on two categories $i<j<k$ strictly straddling $j$, then equality holds in \eqref{eq:envelope}; $\theta$ is the mass $q$ places above $j$, and adjacency of $i,j,k$ is not required. The same holds with the arms exchanged.
\item[(ii)] \textup{(Necessity within the two-atom class.)} Among pairs in which each arm has at most two atoms, equality holds \emph{only} in the configurations of \textup{(i)}. In particular equality at $\prem=16/9$ requires $\theta=\tfrac12$, i.e.\ $q$ splitting its mass equally on the two sides.
\item[(iii)] \textup{(A necessary condition in general.)} If equality holds, then neither arm can carry two atoms inside a block of categories on which its own score function is constant and the other arm places no mass.
\end{enumerate}
\end{theorem}

\begin{proof}
(i) is Theorem~\ref{thm:main}(i), whose computation never used adjacency: if $q$ places mass $1-t$ at $i$ and $t$ at $k$ with $i<j<k$, then $a\equiv t$ on $\{j\}$, $b(Y)\sim\mathrm{Bernoulli}(t)$, $\tau=0$, and $\sum_k m_k^3=\{(1-t)^3+1+t^3\}/8$ exactly as before.

(ii) Equality forces $\Psi=0$. By Proposition~\ref{prop:strict}, within the two-atom class this is possible only on the boundary of the square, i.e.\ only if one arm is degenerate. It remains to audit what a boundary point means, since a degenerate arm may collapse the nominal pattern into a different one. Suppose then $p=e_j$, and let $q$ place masses $u$, $q_0$ and $w$ on the categories below, at, and above $j$ (Lemma~\ref{lem:reduce} leaves at most two atoms in $q$, and a power-sum argument, given in the Supporting Information, allows the blocks below and above to be treated as single categories). Then $\theta=w+q_0/2$, $\tau=q_0$, and a direct expansion gives
\begin{equation}\label{eq:bdry}
\Psi\;=\;-\,q_0\,G(u,w),
\qquad
\tfrac{16}{3}\,G(u,w)\;=\;\varphi(u)+\varphi(w)+uw\bigl\{4(u-w)^2+7(u+w)-20\bigr\},
\end{equation}
with $\varphi$ as in \eqref{eq:Sform}. Since $\varphi(t)\ge7t$, and $uw\le(u+w)^2/4$ with $u+w\le1$ gives $20uw\le5(u+w)^2\le5(u+w)$, we get $\tfrac{16}{3}G\ge7(u+w)-20uw\ge2(u+w)$, so $G>0$ unless $u=w=0$. Hence $\Psi=0$ forces $q_0=0$ or $u=w=0$. The second case is $q=e_j=p$, the degenerate null, which has $v(m)=0$ and is excluded by Remark~\ref{rem:degen}. So $q_0=0$: the comparator places no mass at the degenerate category, and $\theta=w\in(0,1)$ then forces $w>0$ and $u=1-w>0$, i.e.\ mass strictly on both sides of $j$. That is exactly the configuration of part~(i). A degenerate arm whose comparator lies wholly to one side \emph{and carries no coincident mass} ($q_0=0$) gives $w\in\{0,1\}$ and hence $\theta\in\{0,1\}$, which is excluded; a comparator supported at $j$ and on one side only has $q_0>0$ and is already covered by the strictness of \eqref{eq:bdry}.

(iii) Suppose $p$ carries mass at two categories $i\ne i'$ with $a_i=a_{i'}$ and $q_i=q_{i'}=0$. Moving all of that mass onto one of them leaves $\theta$, $s$ and $\tau$ unchanged and strictly increases $\sum_k c_k^3$ by strict Schur-convexity of $x\mapsto x^3$, hence strictly increases $\Psi=3(s-\tfrac{\tau}{4})(2+s)-s(8-\sum_kc_k^3)$ whenever $s>0$. Since $\Psi\le0$ everywhere by Theorem~\ref{thm:global}, a configuration with $\Psi=0$ admits no such strict increase, so no such pair of atoms exists. The argument for $q$ is identical.
\end{proof}

\begin{remark}
Part~(iii) and the numerical evidence are enough for the design use we make of the result: a planner who anticipates dispersed arms is far from equality, and the Supporting Information quantifies the distance. Leaking a fraction $\varepsilon$ of the degenerate arm's mass elsewhere costs roughly $2.5\varepsilon$ in premium.
\end{remark}

\subsection{Allocation}\label{sec:alloc}

Two allocation facts complete the theoretical picture. First, for known components the variance-optimal split of \eqref{eq:vtrue} is the Neyman allocation $\lambda^\star=\sigma_1/(\sigma_1+\sigma_2)$ \cite{happ2019}, and the total-sample-size loss from using balanced allocation instead is exactly
\begin{equation}\label{eq:Lr}
L(r)\;=\;\frac{V_{\mathrm{true}}(\tfrac12)}{V_{\mathrm{true}}(\lambda^\star)}\;=\;\frac{2(1+r^2)}{(1+r)^2},
\qquad r=\frac{\sigma_1}{\sigma_2},
\end{equation}
which is flat near its minimum: $L(1.5)=1.04$, $L(2)=1.11$, $L(3)=1.25$, and $L\le2$ always. Moreover, misallocation enters power only as a multiplicative variance inflation, so any procedure that can revise total $N$ (an internal pilot) recovers the power exactly at the cost of $L(r)-1$ extra patients; the \emph{irreversibility} of the randomization ratio therefore has bounded consequences.

Second, the premium itself is not allocation-invariant. Along the extremal family, with $s=\theta(1-\theta)$,
\begin{equation}\label{eq:lam}
\prem(\lambda)\;=\;\frac{12\,\lambda\,s}{\,1-\lambda^3-(1-\lambda)^3(1-3s)\,},
\end{equation}
which at $s=\theta(1-\theta)$, $\theta=0.55$ evaluates to $0.85$, $1.12$, $1.76$, $2.85$, $3.88$ at $\lambda=\tfrac14,\tfrac13,\tfrac12,\tfrac23,\tfrac34$, and diverges as $\lambda\to1$ (the denominator behaves as $3(1-\lambda)$). Worst-case variance understatement and forced unequal allocation toward the low-information arm therefore \emph{compound}: the $16/9$ cap is a balanced-allocation statement. In the empirical panel this interaction never activates because $r$ stays near $1$ (Section~\ref{sec:empirical}), but designs with externally imposed skewed ratios should not rely on the balanced-allocation bound.

\begin{table}[htbp]
\centering
\caption{The envelope $\bar\prem(\theta)=16\theta(1-\theta)/\{2+\theta(1-\theta)\}$ of Theorem~\ref{thm:global}: the ceiling at each anticipated effect, its two conventional readings, and the resulting worst-case sample size $\bar\prem(\theta)N_{\mathrm{plan}}$. The constant $16/9$ applies only as $\theta\to\tfrac12$. The $\theta=0.50$ column illustrates the limiting variance-ratio ceiling only: the usual effect-based formula is singular at $\theta=\tfrac12$, so $N_{\mathrm{plan}}$ there is a hypothetical baseline rather than a calculated design. Sample sizes are rounded up to the nearest even total.}
\label{tab:envelope}
\small
\begin{tabular}{lcccccc}
\toprule
anticipated $\theta$ & $0.50$ & $0.55$ & $0.60$ & $0.65$ & $0.70$ & $0.75$\\
\midrule
ceiling $\bar\prem(\theta)$ & $16/9$ & $1.762$ & $12/7$ & $1.634$ & $1.520$ & $1.371$\\
excess over calculated $N$ & $+77.8\%$ & $+76.2\%$ & $+71.4\%$ & $+63.4\%$ & $+52.0\%$ & $+37.1\%$\\
variance understated by & $43.8\%$ & $43.2\%$ & $41.7\%$ & $38.8\%$ & $34.2\%$ & $27.1\%$\\
\midrule
$N_{\mathrm{plan}}=200$ & $356$ & $354$ & $344$ & $328$ & $306$ & $276$\\
$400$ & $712$ & $706$ & $686$ & $654$ & $610$ & $550$\\
$800$ & $1424$ & $1410$ & $1372$ & $1308$ & $1218$ & $1098$\\
\bottomrule
\end{tabular}
\end{table}

\section{Realized premia under clinical ordinal alternatives}\label{sec:empirical}

\subsection{Design and results}

The bound holds for all ordinal distributions under balanced allocation; the premium's actual size for a real trial population is an empirical matter.

\subsection*{Real trial data}

Table~\ref{tab:realdata} reports the premium for every pair of arms we were able to extract from published trials with an ordinal endpoint, across three disease areas and sample sizes spanning $41$ to $3035$. No parametric outcome model enters these rows: each premium is the empirical plug-in value from an observed arm-level table for a completed randomized comparison.

\begin{table}[h]
\centering
\caption{Realized planning premium for extracted arm-level tables from published randomized trials. $\bar\prem(\theta)$ is the ceiling of Theorem~\ref{thm:global} at the observed effect. $L(r)$ is the balanced-allocation loss \eqref{eq:Lr}; $D_{\mathrm{rel}}$ is the relative pooled-tie distortion \eqref{eq:pool}.}
\label{tab:realdata}
\small
\begin{tabular}{@{}llrrrrrrr@{}}
\toprule
trial & endpoint & $K$ & $n$ & $\theta$ & $\prem$ & $\bar\prem(\theta)$ & $L(r)$ & $D_{\mathrm{rel}}$\\
\midrule
IST-3 \cite{ist3} & Oxford Handicap, 6\,mo & $7$ & $3035$ & $0.515$ & $0.9997$ & $1.776$ & $1.001$ & $0.03\%$\\
Auranofin RA \cite{bombardier1986} & self-assessment, 1\,mo & $5$ & $299$ & $0.555$ & $1.012$ & $1.759$ & $1.038$ & $1.35\%$\\
 & self-assessment, 3\,mo & $5$ & $296$ & $0.587$ & $0.960$ & $1.730$ & $1.001$ & $0.44\%$\\
 & self-assessment, 5\,mo & $5$ & $293$ & $0.593$ & $0.949$ & $1.723$ & $1.001$ & $0.23\%$\\
Shoulder-tip pain \cite{lumley1996} & pain score, day 2 & $5$ & $41$ & $0.163$ & $0.531$ & $1.022$ & $1.093$ & $13.7\%$\\
\bottomrule
\end{tabular}
\end{table}

\paragraph{Provenance.} We state sources exactly and distinguish what we verified from what we did not. The IST-3 counts were read from Table~2 of the primary publication \cite{ist3}; arm totals reproduce the published $n_1=1520$ and $n_2=1515$. The shoulder-tip pain counts were transcribed from Table~1 of \cite{schuurhuis2025} and then \emph{independently confirmed}, category by category and in both arms, against the \texttt{shoulder} data object distributed in the \texttt{nparLD} package \cite{nparld}, whose evening-of-day-2 measurement reproduces the table exactly with group sizes $22$ and $19$. The arthritis counts were computed by us from the patient-level \texttt{arthritis} data object distributed in the \texttt{multgee} package \cite{multgee}, cross-tabulating the five-level self-assessment against treatment at each follow-up; the object contains $302$ patients ($149$ versus $153$) with $3$, $6$ and $9$ missing responses at one, three and five months, giving the complete-case totals $299$, $296$ and $293$ shown, consistent with the complete-case sample sizes reported in analyses of these data \cite{touloumis2013}. We did \emph{not} independently reconcile the arthritis counts against the original clinical report\cite{bombardier1986}, which we could not obtain; those rows should be read as software-dataset reproductions rather than primary-source extractions.

Five is a small number for a structural reason: published trials with ordinal endpoints rarely report the complete arm-level cross-tabulation that exact evaluation of $\prem$ requires, the conventional report being a dichotomized proportion, a common odds ratio, or a bar chart without counts. Table~\ref{tab:realdata} contains every dataset we could locate with sufficient information. It is a convenience sample of what is publicly recoverable, not a systematic review, and is plausibly biased toward datasets that have become methodological benchmarks.

Four of the five are conservative and one---the one-month arthritis assessment---is anti-conservative by $1.2\%$. The pain study is the least comfortable row: with $n=41$ and a large effect ($\theta=0.163$) it has the most unbalanced components of the set ($r=1.88$, so balanced allocation costs $9.3\%$) and by far the largest pooled-tie distortion ($13.7\%$, the second-order term of Proposition~\ref{prop:pool} being large because the arms differ a great deal). Its premium of $0.53$ means the calculation would have over-sized substantially. Small trials with large effects are where the companion quantities, not the premium, deserve attention.

\subsection*{Generated alternatives}

The control distribution for the scenario panel is likewise not stylized. We extracted the arm-level outcome table of the third International Stroke Trial (IST-3) \cite{ist3}, a randomized trial of intravenous rt-PA in $3035$ patients with acute ischaemic stroke whose primary outcome was the Oxford Handicap Score (a seven-level variant of the modified Rankin Scale) at six months. The published counts, for OHS $0,\dots,5$ and death, are $116,204,214,193,140,246,407$ in the control arm ($n=1520$) and $138,225,191,235,115,203,408$ under rt-PA ($n=1515$). From the real control arm we generated alternatives under three mechanisms---proportional cumulative odds; bounded non-proportional alternatives adding independent category-specific cumulative-logit perturbations up to $\Delta\in\{0.2,0.4,0.6\}$; and an adversarial crossing family shifting survivors toward favourable categories while transferring up to $0.15$ additional mass to death---and repeated the exercise on four stylized control shapes (severe large-vessel-occlusion, mild stroke, moderate, high-mortality) as a check that nothing depends on the single extracted population. Restricting to $\theta\in[0.55,0.70]$, the range these trials are powered for, leaves $4{,}711$ scenarios.

\begin{figure}[t]
\centering
\includegraphics[width=.8\textwidth]{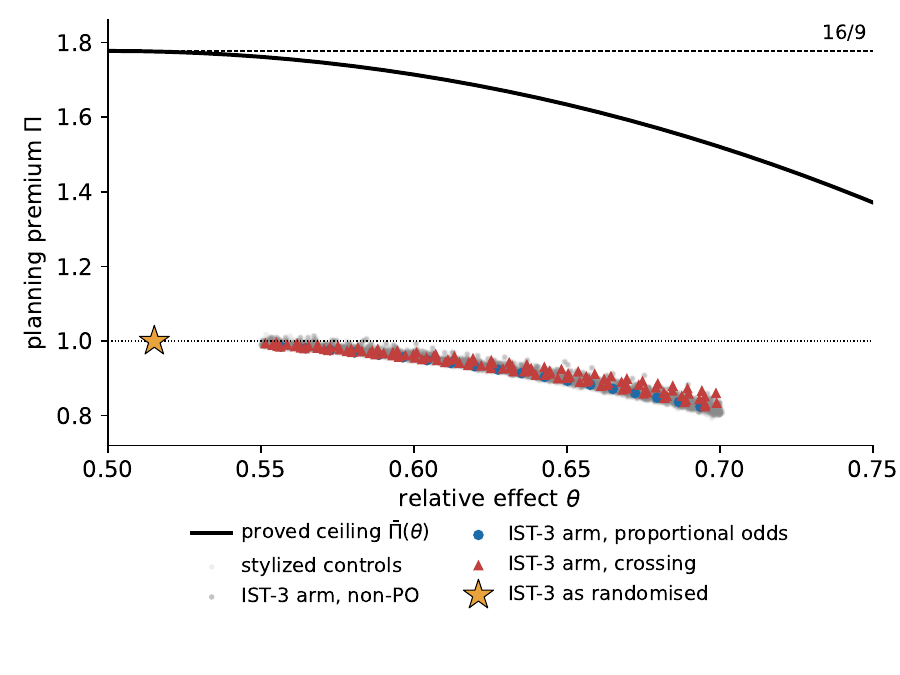}
\caption{The proved ceiling $\bar\prem(\theta)$ against realized premia. Grey points use the extracted IST-3 control arm (light grey: stylized controls); the star is the trial as randomised.}
\label{fig:envelope}
\end{figure}

For the observed IST-3 pair the premium is $0.9997$: the null-calibrated calculation would have been accurate to three decimal places, with $\sigma_1/\sigma_2=0.954$, allocation loss $1.0006$, and pooled-tie distortion $0.03\%$. Across the $4{,}711$ generated scenarios (Figure~\ref{fig:envelope}; full breakdown by mechanism in the Supporting Information) the premium never exceeded $1.019$, with median $0.926$ and $48$ of $4{,}711$ configurations anti-conservative, so the calculation was mildly \emph{conservative} across these scenarios. That is directionally compatible with the pooled-tie effect of Proposition~\ref{prop:pool}, though that proposition alone does not determine the sign of $\prem-1$. Even the worst anti-conservative case was a $1.9\%$ variance understatement, against a worst-case understatement of $43.2\%$ at the same effect size---equivalently a $1.9\%$ versus $76.2\%$ excess over the calculated sample size. The component ratio stayed near balance in every scenario examined ($r\in[0.60,1.35]$, median $0.997$), so $L(r)\le1.064$ and the compounding of \eqref{eq:lam} never activated; the pooled-tie distortion \eqref{eq:pool} ranged over $[0.05\%,4.7\%]$, always conservative and growing with effect size as its second-order character predicts.

\subsection{Non-shift treatment mechanisms}\label{sec:nonshift}

The scenarios above are shift-type, which is the obvious objection: the extremal configuration is not a shift, so a panel of shifts cannot say much about how close practice comes to it. We therefore evaluated, all on the extracted IST-3 control arm, four mechanism families that deliberately violate the shift assumption---a responder subgroup piling up in one category (``new mode'', $39$ configurations); mass moved symmetrically to both extremes (``dispersion only'', $11$); benefit for most with toxicity for a subgroup ($44$); and outcomes compressed at one end of the scale or a cured fraction (ceiling, floor and cure mixture, $47$)---together with a fifth family that perturbs the control arm toward a point mass while the comparator splits to either side ($18$).

There are three findings; the full breakdown is in the Supporting Information. First, non-shift mechanisms do push the premium above one more often and further than shift-type ones: the largest value outside that family is $1.091$, against $1.019$ across the entire shift-type panel. Second, the direction is systematic rather than sporadic. A mechanism that increases dispersion without moving location---treatment moves mass symmetrically toward both ends---was anti-conservative in all $11$ configurations examined, and a mechanism combining benefit with toxicity in $23$ of $44$. This is the expected consequence of the geometry: both push mass outward from the centre while leaving the pooled distribution comparatively concentrated, which is a weak version of the extremal configuration. Ceiling, floor and cure mechanisms move in the opposite direction and are uniformly conservative.

Third, and most usefully for a planner, the fifth family shows the bound is not decorative. Making the control arm progressively closer to a point mass while the comparator splits to either side, the premium rises through $1.23$, $1.58$ and reaches $1.752$ at a control arm placing $99\%$ of its mass in one category---$98.5\%$ of the ceiling $\bar\prem(\tfrac12)$. The worst case is reachable, by distributions that are recognizable in advance, which is what makes the design check of Section~\ref{sec:practice} worth performing.

\subsection{From asymptotic variance to realized power}\label{sec:finite}

The premium is a ratio of asymptotic variances, so it must be checked that it predicts finite-sample power at the sample sizes these designs imply. We planned $N$ by the tie-corrected null-calibrated formula for $80\%$ power at $\alpha=0.05$ two-sided, then simulated the studentized Brunner--Munzel test\cite{brunner2000} at that $N$ ($3{,}000$--$4{,}000$ replicates, seven categories), comparing realized power with the first-order prediction from the premium alone, $\Phi\{(z_{1-\alpha/2}+z_{1-\beta})/\sqrt{\prem}-z_{1-\alpha/2}\}$.

Across all eight settings---three IST-3 shifts, a new-mode alternative, the extremal family, and shifts calibrated so that $N_{\mathrm{plan}}$ equals $100$, $200$ and $500$---the prediction was accurate to within $2.1$ percentage points, the residual being a small negative bias attributable to the studentized test at moderate $N$ rather than to the premium (full table in the Supporting Information). At the extremal configuration the agreement is $0.567$ against $0.560$ predicted: where the premium is large, essentially the entire power deficit is accounted for by it. The asymptotic quantity therefore transfers to the sample sizes at which these trials are actually run, from $N=100$ upward.

\section{Discussion}\label{sec:disc}

\subsection{Summary of the results}

The premium \eqref{eq:prem} isolates the variance substitution described in Section~\ref{sec:intro}, which is forced rather than chosen; the aim here is to price it, not to remove it. Lemma~\ref{lem:null} and Theorems~\ref{thm:global}--\ref{thm:equality} establish that the substitution is exact at any non-degenerate null, that its cost is bounded by the envelope, and that within the two-atom class the bound is attained only on the extremal family.

Although that configuration is highly structured, the bound is not merely theoretical: in the worst case the required sample size is substantially larger than the conventional calculation, and Section~\ref{sec:finite} confirms the deficit materializes as lost power. It requires a perfectly homogeneous control arm while treatment splits an identical population between categories on opposite sides of the control arm's category, leaving nobody at baseline: the competing requirements described in Section~\ref{sec:intro}, met by the extremal family and, within the two-atom class, only by it. Prognostic heterogeneity keeps real arms dispersed, and treatments move mass \emph{along} the scale rather than symmetrically \emph{away} from an interior point. The protection does not come from marginal regularity---concentration caps and even unimodality of both arms leave most of the worst case intact (Supporting Information)---but from the coupling of the arms through the treatment mechanism. Section~\ref{sec:nonshift} shows that mechanisms which break that coupling, in particular those widening dispersion without moving location, are systematically anti-conservative, and that a near-degenerate control arm reaches $98.5\%$ of the ceiling.

\subsection{Scope of the bound}\label{sec:notcover}

One example is worth stating explicitly, because the constant invites over-reading. Suppose a planner anticipates $\theta=0.75$ with dispersed arms, computes $N$ from the tie-corrected Noether formula, and inflates by $16/9$ ``to be safe.'' The premium for such a configuration is near $1$, so the inflation is almost pure waste---but more importantly, it does not protect against the error that actually threatens this design. Noether's formula is a local approximation whose accuracy degrades through the $(\theta-\tfrac12)^{-2}$ factor as the effect grows, and that error is present even when the variance is specified exactly. Our premium holds the effect fixed and asks only what the variance substitution costs; it is silent on the rest of the formula.

The same boundary applies elsewhere. The bound is a balanced-allocation statement---\eqref{eq:lam} shows it fails badly as $\lambda\to1$. It concerns asymptotic variance, not finite-sample behaviour of any particular test. And it is a planning-stage result: it says nothing about Type~I error under contamination, about estimand mismatch between rank-based and mean-based questions, or about whether $\theta$ is the right target in the first place. A bounded planning premium is not a licence to stop thinking about the analysis.

\subsection{Using the bound in practice}\label{sec:practice}

The constant $16/9$ is not a recommended inflation factor: Section~\ref{sec:empirical} shows realized premiums close to one, so routinely inflating by $77.8\%$ would waste participants. The bound is a \emph{design sensitivity certificate}, converting an unquantified worry about unknown outcome shapes into a number that can be reported. It has three uses.

\emph{Reporting a worst case.} Compute $N_{\mathrm{plan}}$ as usual, then report $\bar\prem(\theta)N_{\mathrm{plan}}$ as the largest requirement consistent with any ordinal configuration at the anticipated effect (Table~\ref{tab:envelope}). Use the envelope at the planned $\theta$ rather than the flat $16/9$, which needlessly overstates the exposure of any trial powered for $\theta>\tfrac12$. Note also that the ceiling depends on the convention: $16/9$ with the tie correction, $3/2$ without (Corollary~\ref{cor:uncorrected}), so protocols should state which null variance was used.

\emph{Licensing the conventional calculation.} Where prior knowledge indicates a dispersed control arm and a shift-type mechanism, a protocol can state that the configurations producing the worst case are excluded by the anticipated distributions and that realized premiums in comparable settings sit near one. This is the more common use: the theorem licenses \emph{not} inflating, and does so with a guarantee rather than an assurance.

\emph{Auditing a completed trial.} $\prem$ is computable from the observed arm-level table, as in Table~\ref{tab:realdata}. Reporting $\widehat\prem$ tells readers how accurate the original variance assumption was, and accumulating these across trials in a disease area would build the evidence base this paper lacks.

\paragraph{When to distrust the calculation.} Two design features move a trial toward the extremal geometry or defeat the premise of the bound, and either should trigger component-based planning by the method of Happ et al.\cite{happ2019} using pilot or registry data. The first is an arm anticipated to be nearly degenerate---a severe floor or ceiling, most patients expected in a single category---faced by a comparator that plausibly splits patients to \emph{both} sides of that category. The second is an externally forced allocation ratio materially beyond $2{:}1$, which voids the cap by \eqref{eq:lam}. Absent both, and with a dispersed anticipated control distribution and a shift-type mechanism, the conventional calculation was accurate to within a few percent in every scenario we examined, more often conservative than not, with a $5\%$ inflation of $N$ covering all of them. That last statement is an empirical finding about the settings of Section~\ref{sec:empirical}, not a guarantee; the guarantee is the ceiling.

\paragraph{A procedure.} For a trial with an ordinal endpoint to be analysed by a rank-based test, we suggest the following four steps, which add a few minutes to a calculation that is already being done.

\begin{enumerate}
\item Compute $N_{\mathrm{plan}}$ as usual, and record whether the tie-corrected or uncorrected null variance was used, since the ceiling differs ($16/9$ against $3/2$).
\item Inspect the anticipated control distribution. If no category is expected to hold more than about half the patients, and the anticipated mechanism is a shift rather than a symmetric widening, proceed with $N_{\mathrm{plan}}$; the evidence of Sections~\ref{sec:empirical} and~\ref{sec:nonshift} is that the error is then within a few percent, usually conservative.
\item If instead one arm is anticipated to be nearly degenerate, or the mechanism plausibly moves mass outward to both ends of the scale, or allocation is forced materially beyond $2{:}1$, plan from the variance components directly\cite{happ2019}, using pilot or registry data; failing that, size the trial at $\bar\prem(\theta)N_{\mathrm{plan}}$ from Table~\ref{tab:envelope}.
\item Report both $N_{\mathrm{plan}}$ and the worst-case $\bar\prem(\theta)N_{\mathrm{plan}}$ in the protocol, and report the realized $\widehat\prem$ on completion.
\end{enumerate}

\paragraph{A protocol summary.} These fit into four lines that can be tabulated in a statistical analysis plan. Take a seven-category endpoint with anticipated pooled distribution
\[
(0.10,\;0.16,\;0.18,\;0.18,\;0.15,\;0.11,\;0.12),
\]
anticipated $\theta=0.60$, $90\%$ power and $\alpha=0.05$ two-sided:

\begin{center}
\small
\begin{tabular}{lr}
\toprule
Conventional sample size $N_{\mathrm{plan}}$ & $344$\\
Worst-case premium $\bar\prem(0.60)$ & $12/7\approx1.714$\\
Worst-case sample size & $590$\\
Realized premium $\widehat\prem$ (reported on completion) & ---\\
\bottomrule
\end{tabular}
\end{center}

\noindent Sample sizes are rounded up to the nearest even total throughout, so that $344\times12/7=589.7$ becomes $590$; the same rule is used in Table~\ref{tab:envelope}. The first three lines are computable at the design stage from quantities the protocol already specifies; the fourth is filled in at reporting. Nothing about how the sample size is computed changes.

\subsection{Implications for teaching}

Introductory treatments often pair ``nonparametric tests need no assumptions'' with rituals of assumption-checking for parametric ones. The accurate statement, now with exact constants attached, is that rank procedures are assumption-free for \emph{validity} (with studentized implementations \cite{brunner2000,neubert2007}) but not for \emph{planning}---and that the planning error from ignoring shape is bounded by a single constant and was small in the clinical and generated scenarios examined here, a lesson consonant with recent arguments for defaulting to robust procedures rather than pretesting \cite{delacre2017,zimmerman2004}.

\subsection{Limitations}\label{sec:limits}

Five limitations bound the claims, and the first two are the ones on which the paper is most exposed.

\emph{Novelty.} The contribution is a single bound, and part of the machinery is classical: the cap on the numerator is Birnbaum--Klose \cite{birnbaum1957,rustagi1961}, the kernel-variance identity is Bamber's \cite{bamber1975,brunner2025}, and the idea of studying a variance \emph{ratio} appears in \cite{schuurhuis2025}, with a different denominator and for a different purpose. The imperfection of the null substitution is itself known and is patched heuristically in \cite{poehlmann2024}. What we claim is the sharp comparison against the tie-corrected null variance, its envelope, the equality characterization, and Corollary~\ref{cor:uncorrected}. The author has read the primary sources cited above rather than relying on secondary description, but a literature this scattered admits no exhaustive search; if an equivalent envelope exists, the paper reduces to an exposition, and a reader who knows of one should say so.

\emph{The equality set.} Theorem~\ref{thm:equality} characterizes equality only within the two-atom class, plus a necessary block-sparsity condition in general. Whether equality can occur elsewhere is open; Appendix~\ref{app:eqnum} explains exactly where the natural arguments break and what the numerical evidence says. This does not affect the bound itself, which is proved in full generality, but it does mean the phrase ``attained only when one arm is degenerate'' should be read with that qualification.

\emph{Length of the case analysis.} The proof of Theorem~\ref{thm:global} is elementary and requires no computer verification---the symmetry reduction leaves six orbits, three of which are immediate and three of which reduce to the two displayed identities---but it is still a case analysis, and a shorter or more conceptual route to the envelope would be welcome. The constrained maxima of the Supporting Information remain certified only as lower bounds, and sharp suprema under unequal allocation $\lambda\neq\tfrac12$ are open.

\emph{Empirical scope.} Section~\ref{sec:empirical} reports five extracted arm pairs from three disease areas, a scenario panel built on the IST-3 control arm and four stylized shapes, and five non-shift mechanism families. Five real comparisons are more than one but still few, and they are concentrated in stroke, rheumatology and post-surgical pain; the generated alternatives, shift-type and otherwise, are model constructions rather than observed trials. The empirical magnitudes are therefore evidence about the settings examined, not a general statement about ordinal endpoints, and the mismatch between a theorem covering every ordinal distribution and a panel drawn from a few disease areas should be read as a limit on the empirical claims, not on the bound. Disease areas with heavier tie structures or routinely unbalanced designs could behave differently, and further extraction is the obvious next step.

\emph{Asymptotics.} The premium is an asymptotic variance ratio. Section~\ref{sec:finite} shows it transfers accurately to realized power at the sample sizes these designs imply, but that check covers three configurations at moderate $N$; small-trial regimes, where the studentized permutation test \cite{neubert2007} is preferable, need separate study.

\emph{Scope of the object.} Everything here concerns the \emph{planning} stage. Bounded planning premia say nothing about analysis-stage robustness---Type~I error under heteroscedasticity, contamination, or estimand mismatch between rank-based and mean-based questions are separate problems and are not made small by these results.

\medskip
\noindent The contribution is therefore a sharp worst-case guarantee for the planning-variance substitution used in standard rank-based sample-size calculations, together with an explicit extremal configuration, a characterization of equality within the two-atom class, and evidence that this worst case is distant from the clinical scenarios we examined. Rank tests are distribution-free for validity but not for planning; the price of that distinction is not merely finite but exactly known.

\appendix

\section{The equality set: what is not proved}\label{app:eqnum}

\begin{remark}[What is not proved]\label{rem:eqgap}
Part~(ii) is confined to the two-atom class, and we do not prove that (i) exhausts the equality set over all of $\Delta_{K-1}^2$. The obstruction is that equality in Lemma~\ref{lem:kernel} requires only that the kernel $h$ be additive on $\operatorname{supp}(p)\times\operatorname{supp}(q)$, which does not force a degenerate arm: for $\operatorname{supp}(p)=\operatorname{supp}(q)=\{y_1,y_2\}$ the kernel matrix is $\binom{1/2\;\;1}{0\;\;1/2}$, additive with both variance components positive. Nor does Lemma~\ref{lem:reduce} help, since it asserts that \emph{some} maximizer is two-atomic, not that every maximizer is. Numerical search finds no equality configurations outside (i), with margins reported in Appendix~\ref{app:eqnum}; we therefore expect (i) to be the full equality set but state it as an expectation.
\end{remark}

Complementing Appendix~\ref{app:eqnum}: searching for configurations that approach the envelope while forcing both arms to carry a second atom of mass at least $0.01$ leaves an envelope deficit of at least $4.9\times10^{-3}$; against a degenerate arm, giving $q$ three or more atoms leaves a deficit of at least $1.5\times10^{-3}$; and requiring $\tau\ge0.02$ leaves a deficit of at least $4.7\times10^{-2}$. The two additive configurations displayed in Appendix~\ref{app:eqnum} fall well short, with deficits $0.147$ and $0.754$.

\section{Reproducibility}\label{sec:verify}

The two polynomial identities \eqref{eq:Fdecomp} and \eqref{eq:Sform}, and every inequality in the proof, are re-checked symbolically by \texttt{handproof.py}, which also confirms Lemma~\ref{lem:sym} numerically in dimensions $2$--$8$. Independently of Lemma~\ref{lem:reduce}, $\Psi\le0$ was verified by exhaustive enumeration over full simplices on rational mass grids---$(K,\text{grid})=(3,\tfrac18),(4,\tfrac16),(5,\tfrac15),(6,\tfrac14),(7,\tfrac13)$, some $48{,}000$ exactly evaluated pairs in integer arithmetic---with maximum exactly $0$, attained only at coincident point masses (\texttt{verify\_reduction.py}).

The global numerical search and its constrained variants are described in the Supporting Information. Every figure and table in the paper is regenerated by a single driver script: Table~\ref{tab:envelope} and the Supporting Information from \texttt{search.py} and \texttt{ccurve.py}, Figure~\ref{fig:envelope} and the Supporting Information tables from \texttt{empirical\_real.py} (which reads the IST-3 counts given in Section~\ref{sec:empirical} and writes \texttt{empirical\_real\_results.csv}), and Section~\ref{sec:finite} from \texttt{finite\_sample.py}. Random seeds are fixed in each script.

\section{The identities for orbits 4--6}\label{app:ident}

Every step below is an identity between polynomials with integer or rational coefficients, verifiable by expansion with no computer algebra beyond that, followed by elementary bounds on the unit square. The accompanying script re-checks each identity symbolically, but the verification does not depend on it.

\emph{Orbit 6 ($PQ{<}PQ$).} Here $\Psi=-\tfrac{3}{16}F$ with $F$ of bidegree $(4,4)$, and
\begin{equation}\label{eq:Fdecomp}
F\;=\;7\bigl[\alpha(1-\alpha)+\beta(1-\beta)\bigr]\;+\;(\alpha-\beta)^2\,R,
\qquad
R=(u-1)(7u-8)+2\alpha\beta,\quad u=\alpha+\beta,
\end{equation}
an identity verified by expansion. If $u\le1$ or $u\ge\tfrac87$ then $(u-1)(7u-8)\ge0$ and $R\ge0$. If $1<u<\tfrac87$ then $\alpha,\beta\le1$ forces $\alpha\beta\ge u-1>0$, so
$R\ge (u-1)(7u-8)+2(u-1)=(u-1)(7u-6)>0$. Hence $R\ge0$, $F\ge0$, and $\Psi\le0$.

\emph{Orbits 4 and 5 (one shared position).} These are the cases where the two arms overlap, so $\tau>0$ and the tie term contributes; they are handled together because, after the substitutions below, they produce the same functional form.

After removing manifestly signed factors, $P{<}PQ{<}Q$ gives $\Psi=-\tfrac{3}{16}\beta\bar\alpha H$ and $PQ{<}P{<}Q$ gives $\Psi=-\tfrac{3}{16}\alpha\beta\,\widetilde G$, with $H,\widetilde G$ of bidegree $(3,3)$; it suffices that $H\ge0$ and $\widetilde G\ge0$ on $[0,1]^2$. Each vanishes at exactly one corner, and translating that corner to the origin puts both into the \emph{same} shape: writing $(\alpha,\beta)=(x,1-y)$ for $H$ and $(\alpha,\beta)=(1-x,1-y)$ for $\widetilde G$,
\begin{equation}\label{eq:Sform}
S(x,y)\;=\;\varphi(x)+\varphi(y)+x\,y\,W(x,y),
\qquad
\varphi(t)=t\,(7+8t-7t^2),
\end{equation}
with $W=W_H$ and $W=W_{\widetilde G}$ respectively. The translation is the natural move: each cofactor is strictly positive except at one corner of the square, so shifting that corner to the origin exposes the vanishing behaviour, and what remains is a statement that a positive combination of $x$ and $y$ dominates their product. Three elementary facts finish the proof.

\emph{(a)} $\varphi(t)-7t=t^2(8-7t)\ge0$ on $[0,1]$, so $\varphi(x)+\varphi(y)\ge7(x+y)$.

\emph{(b)} $x+y-2xy=x(1-y)+y(1-x)\ge0$ on $[0,1]^2$, so $7(x+y)\ge14xy$.

\emph{(c)} $W\ge-14$ on $[0,1]^2$ in both cases. For $W_H=-x^2y^2-x^2y-xy^2+9x^2+18xy+9y^2-33x-33y+25$, the three cubic terms satisfy $-x^2y^2-x^2y-xy^2\ge-3xy$ (each factor is at most $1$), and $-3xy\ge-\tfrac34u^2$ with $u=x+y$, so
$W_H\ge 9u^2-\tfrac34u^2-33u+25=\tfrac{33}{4}u^2-33u+25$, which is decreasing on $[0,2]$ and equals $-8$ at $u=2$; hence $W_H\ge-8$. For $W_{\widetilde G}=-x^2y^2-x^2y+9x^2-9xy^2+18xy-9x-15y^2+27y-11$, use $-x^2y^2-x^2y\ge-2y$ and $-9xy^2\ge-9y^2$ to get $W_{\widetilde G}\ge 9x^2+18xy-9x-24y^2+25y-11$, which is concave in $y$ and so minimized at $y\in\{0,1\}$: at $y=0$ it is $9x^2-9x-11\ge-\tfrac{53}{4}$, at $y=1$ it is $9x^2+9x-10\ge-10$. Hence $W_{\widetilde G}\ge-\tfrac{53}{4}$.

Combining, $S\ge7(x+y)-14xy\ge0$ in both cases, so $\Psi\le0$.

\section{The six orbits}\label{app:cert}

By Lemma~\ref{lem:reduce} it suffices to bound $\Psi$ of \eqref{eq:psi} over pairs with at most two atoms per arm; write the $p$-masses as $(\alpha,\bar\alpha)$ and the $q$-masses as $(\beta,\bar\beta)$ in increasing position, with single-atom arms the edges $\alpha,\beta\in\{0,1\}$. Only the order/coincidence pattern of the at most four positions matters. There are thirteen such patterns---six with four distinct positions, six with one shared position, one with both shared---and by Lemma~\ref{lem:sym} they fall into the six orbits below, each handled in the proof of Theorem~\ref{thm:global}. We write $PQ$ for a shared position.

\begin{center}
\small
\begin{tabular}{llll}
\toprule
\# & orbit & representative $\Psi$ & argument\\
\midrule
1 & $P{<}P{<}Q{<}Q$, $Q{<}Q{<}P{<}P$ & $0$ & $\theta$ constant, $s=\tau=0$\\
2 & $P{<}Q{<}Q{<}P$, $Q{<}P{<}P{<}Q$ & $-3\alpha\bar\alpha\beta\bar\beta$ & manifest\\
3 & $P{<}Q{<}P{<}Q$, $Q{<}P{<}Q{<}P$ & $-3\alpha\bar\alpha\beta\bar\beta(1-\bar\alpha\beta)(1+2\beta-\alpha-\alpha\beta)$ & manifest\\
4 & $P{<}PQ{<}Q$, $Q{<}PQ{<}P$ & $-\tfrac{3}{16}\beta\bar\alpha H$ & decomposition \eqref{eq:Sform}\\
5 & $PQ{<}P{<}Q$ and $3$ others & $-\tfrac{3}{16}\alpha\beta\widetilde G$ & decomposition \eqref{eq:Sform}\\
6 & $PQ{<}PQ$ & $-\tfrac{3}{16}F$ & decomposition \eqref{eq:Fdecomp}\\
\bottomrule
\end{tabular}
\end{center}

\noindent The bidegree-$(3,3)$ cofactors of orbits 4 and 5 are
\begin{align*}
H\;=\;&\alpha^3\beta^3-4\alpha^3\beta^2-4\alpha^3\beta+\alpha^2\beta^3+15\alpha^2\beta^2-8\alpha^2\\
&{}-9\alpha\beta^3-6\alpha\beta^2+14\alpha\beta+8\alpha+7\beta^3-13\beta^2-2\beta+8,\\[2pt]
\widetilde G\;=\;&-\alpha^3\beta^3+4\alpha^3\beta^2+4\alpha^3\beta+12\alpha^2\beta^3-21\alpha^2\beta^2-12\alpha^2\beta+8\alpha^2\\
&{}-36\alpha\beta^3+48\alpha\beta^2+10\alpha\beta-24\alpha+32\beta^3-44\beta^2-4\beta+24,
\end{align*}
so that $\Psi=-\tfrac3{16}\beta\bar\alpha H$ on orbit~4 and $\Psi=-\tfrac3{16}\alpha\beta\widetilde G$ on orbit~5. The bidegree-$(4,4)$ cofactor of orbit 6 is
\begin{align*}
F\;=\;&7\alpha^4+2\alpha^3\beta-15\alpha^3-18\alpha^2\beta^2+15\alpha^2\beta+\alpha^2+2\alpha\beta^3\\
&{}+15\alpha\beta^2-16\alpha\beta+7\alpha+7\beta^4-15\beta^3+\beta^2+7\beta,
\end{align*}
with $\Psi=-\tfrac3{16}F$. Substituting $(\alpha,\beta)=(x,1-y)$ in $H$ and $(\alpha,\beta)=(1-x,1-y)$ in $\widetilde G$ produces \eqref{eq:Sform} with
\begin{align*}
W_H&=-x^2y^2-x^2y-xy^2+9x^2+18xy+9y^2-33x-33y+25,\\
W_{\widetilde G}&=-x^2y^2-x^2y+9x^2-9xy^2+18xy-9x-15y^2+27y-11 .
\end{align*}

\section{Global numerical search and interpretable constraints}\label{app:si:search}

Theorem~\ref{thm:global} settles the supremum, so the search reported here is
corroboration rather than evidence; we retain it because it was run before the
proof existed and because the constrained variants carry information the theorem
does not. We searched $\Delta_6\times\Delta_6$ ($K=7$, the modified Rankin
dimension) in the three layers of Algorithm~\ref{alg:search}.

\begin{algorithm}[htbp]
\caption{Global search for $\sup\prem$ over $\Delta_{K-1}^2$}
\label{alg:search}
\begin{algorithmic}[1]
\State \textbf{Layer 1 (enumeration).} Enumerate all $p$ with at most $3$
occupied categories on a mass grid of step $0.1$; form all ordered pairs
$(p,q)$; evaluate $\prem$ in closed form; retain the top $10^3$ and all pairs
within $10^{-3}$ of the running maximum.
\State \textbf{Layer 2 (refinement).} From each retained pair, run sequential
quadratic programming on $(p,q)$ with simplex constraints, both unconstrained in
$\theta$ and subject to $\theta\ge\theta_0$ for each $\theta_0$ in the reporting
grid; optionally add $\max_k p_k\le c$, $\max_k q_k\le c$, or unimodality of
both arms.
\State \textbf{Layer 3 (multistart).} Repeat Layer 2 from $300$
Dirichlet-random interior starts with dispersion drawn from
$\mathrm{U}(0.05,3)$.
\State \textbf{Report} the maximum over all layers and the argmax support pattern.
\end{algorithmic}
\end{algorithm}

Layer~1 evaluated $2{,}119{,}936$ ordered pairs. No configuration exceeded
$16/9$; every refined optimum coincided with the family of Theorem~\ref{thm:main}(i), and the
$\theta$-constrained maxima matched the envelope to five decimals. Unit
tests---the identity $\E\,a(X)=\E\,b(Y)=\theta$ against brute force on $4{,}000$
random pairs, the null-exactness lemma on random null configurations, and the
worked $K=5$ example---pass at machine precision.

\begin{table}[htbp]
\centering
\caption{Largest premium found ($\theta\ge0.55$, $\lambda=\tfrac12$, $K=7$)
under interpretable marginal constraints. Because Layers~2--3 are local
searches, these are certified \emph{lower} bounds on the constrained suprema:
each is attained by an explicit feasible configuration, but local search cannot
certify global optimality.}
\label{tab:ccurve}
\small
\begin{tabular}{lcccccccc}
\toprule
cap $c$ on $\max_k p_k$ and $\max_k q_k$ & $1.0$ & $0.6$ & $0.5$ & $0.4$ & $0.3$ & $0.25$ & $0.20$ & $0.18$\\
\midrule
concentration cap only & $1.762$ & $1.592$ & $1.577$ & $1.536$ & $1.474$ & $1.372$ & $1.156$ & $1.062$\\
cap $+$ unimodality (both arms) & $1.651$ & --- & --- & $1.269$ & $1.168$ & $1.123$ & --- & ---\\
\bottomrule
\end{tabular}
\end{table}

Two structural findings emerge (Table~\ref{tab:ccurve}, Figure~\ref{fig:ccurve}).
First, capping the concentration of \emph{one} arm accomplishes almost nothing:
constraining $\max_k p_k\le0.2$ while leaving $q$ free reduces the maximum only
from $1.762$ to $1.706$, because the optimizer exchanges the roles of the arms.
The mechanism requires \emph{one} near-degenerate arm and does not care which.
Second, symmetric caps at clinically dispersed levels leave most of the worst
case intact ($1.474$ at $c=0.3$), and adding unimodality of both arms still
permits $1.168$. Yet the realized premia of the main text, whose arms satisfy
exactly these marginal properties, never exceed $1.02$. Marginal shape
constraints therefore explain only part of the gap; the binding structure is the
joint one.

\begin{figure}[htbp]
\centering
\includegraphics[width=.8\textwidth]{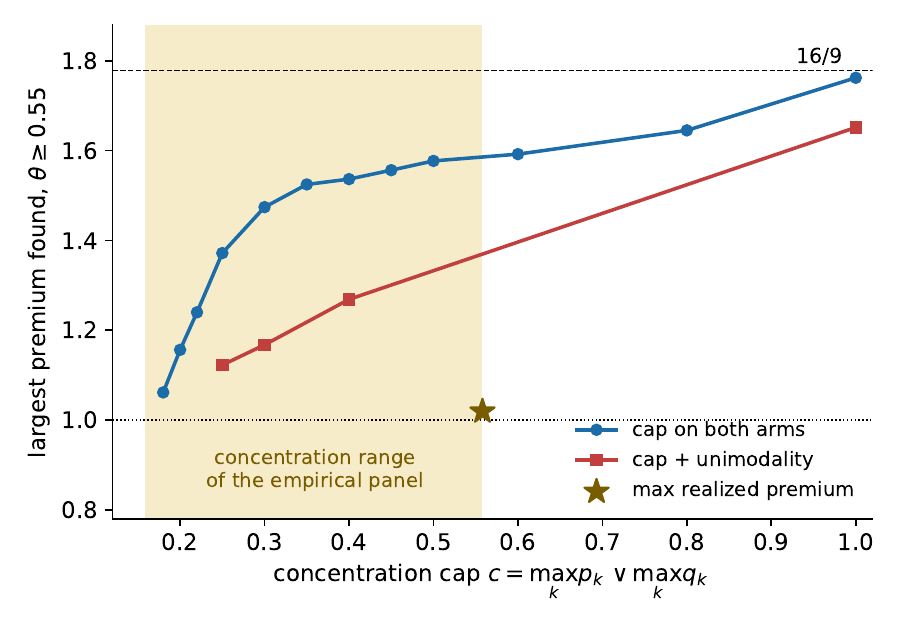}
\caption{Largest premium found under symmetric concentration caps and under caps
plus unimodality of both arms ($\theta\ge0.55$, $\lambda=\tfrac12$, $K=7$;
certified feasible points). The shaded band marks the range of $\max_k$ mass
across the arms of the empirical panel; the star marks the largest premium
realized there.}
\label{fig:ccurve}
\end{figure}

\section{Scenario panel by mechanism}\label{app:si:panel}

Table~\ref{tab:panel} gives the breakdown summarized in Section~\ref{sec:empirical}: $4{,}711$ scenarios with $\theta\in[0.55,0.70]$, generated from the
extracted IST-3 control arm and from four stylized control shapes under
proportional-odds (PO), bounded non-proportional (non-PO), and crossing
mechanisms.

\begin{table}[htbp]
\centering
\caption{Realized planning premium $\prem(\tfrac12)$ for $\theta\in[0.55,0.70]$.
``$\#\,\prem>1$'' counts configurations in which the null-calibrated calculation
under-sizes.}
\label{tab:panel}
\small
\begin{tabular}{llrcccr}
\toprule
control arm & mechanism & $n$ & min & median & max & $\#\,\prem>1$\\
\midrule
IST-3 (extracted) & proportional odds & $19$ & $0.824$ & $0.923$ & $0.985$ & $0$\\
 & bounded non-PO & $2{,}399$ & $0.804$ & $0.926$ & $1.012$ & $27$\\
 & crossing & $114$ & $0.826$ & $0.933$ & $0.997$ & $0$\\
\midrule
stylized (4 shapes) & proportional odds & $80$ & $0.810$ & $0.919$ & $0.986$ & $0$\\
 & bounded non-PO & $1{,}961$ & $0.798$ & $0.925$ & $1.019$ & $12$\\
 & crossing & $138$ & $0.834$ & $0.949$ & $1.019$ & $9$\\
\midrule
\multicolumn{2}{l}{pooled} & $4{,}711$ & $0.798$ & $0.926$ & $1.019$ & $48$\\
\bottomrule
\end{tabular}
\end{table}

\section{Non-shift treatment mechanisms}\label{app:si:nonshift}

Table~\ref{tab:nonshift} gives the breakdown summarized in Section~4.2 of the
main text. All families are built on the extracted IST-3 control arm except the
last, which perturbs the control arm toward a point mass.

\begin{table}[htbp]
\centering
\caption{Premium under treatment mechanisms that are not location shifts.}
\label{tab:nonshift}
\small
\begin{tabular}{llrcccr}
\toprule
mechanism & description & $n$ & min & median & max & $\#\,\prem>1$\\
\midrule
new mode & responder subgroup piles up in one category & $39$ & $0.597$ & $0.969$ & $1.091$ & $13$\\
dispersion only & mass moved symmetrically to both extremes & $11$ & $1.001$ & $1.020$ & $1.070$ & $11$\\
benefit $+$ toxicity & most improve, a subgroup moves to worst & $44$ & $0.815$ & $1.001$ & $1.041$ & $23$\\
ceiling & outcomes compressed at the favourable end & $17$ & $0.312$ & $0.796$ & $0.992$ & $0$\\
floor & outcomes compressed at the unfavourable end & $17$ & $0.401$ & $0.963$ & $0.984$ & $0$\\
cure mixture & fraction cured, remainder unchanged & $13$ & $0.654$ & $0.899$ & $0.998$ & $0$\\
\midrule
near-degenerate & control concentrated in one category, & $18$ & $1.232$ & $1.576$ & $1.752$ & $18$\\
 & comparator split to both sides & & & & &\\
\bottomrule
\end{tabular}
\end{table}

\section{Finite-sample power calibration}\label{app:si:power}

Table~\ref{tab:power} gives the calibration summarized in Section~4.3 of the
main text. $N$ was planned by the tie-corrected null-calibrated formula for
$80\%$ power at $\alpha=0.05$ two-sided; realized power was then simulated for
the studentized Brunner--Munzel test at that $N$ ($3{,}000$--$4{,}000$
replicates, seven categories). The prediction is
$\Phi\{(z_{1-\alpha/2}+z_{1-\beta})/\sqrt{\prem}-z_{1-\alpha/2}\}$.

\begin{table}[htbp]
\centering
\caption{Realized power against the power predicted from the premium. Upper
block: planned $N$ for each scenario. Lower block: shift calibrated so that the
planned $N$ equals $100$, $200$, $500$.}
\label{tab:power}
\small
\begin{tabular}{lccccc}
\toprule
scenario & $\theta$ & $\prem$ & $N$ & predicted & simulated\\
\midrule
IST-3 shift, large & $0.616$ & $0.938$ & $183$ & $0.825$ & $0.808$\\
IST-3 shift, moderate & $0.579$ & $0.971$ & $401$ & $0.811$ & $0.793$\\
IST-3 shift, small & $0.553$ & $0.987$ & $907$ & $0.805$ & $0.798$\\
new mode & $0.607$ & $0.969$ & $221$ & $0.812$ & $0.804$\\
extremal family & $0.550$ & $1.762$ & $883$ & $0.560$ & $0.567$\\
\midrule
calibrated shift & $0.634$ & $0.889$ & $100$ & $0.844$ & $0.823$\\
calibrated shift & $0.588$ & $0.943$ & $200$ & $0.823$ & $0.802$\\
calibrated shift & $0.556$ & $0.977$ & $500$ & $0.809$ & $0.808$\\
\bottomrule
\end{tabular}
\end{table}

\section{Self-contained proof of Theorem 1(ii)}\label{app:si:thm1ii}

Theorem~\ref{thm:main}(ii) states that over all pairs
$(p,q)\in\Delta_{K-1}^2$ in which at least one arm is a point mass,
$\sup\prem=16/9$, attained exactly and only by the family of Theorem~\ref{thm:main}(i) with
$t=\tfrac12$. The result is subsumed by Theorem~2, but the following proof is
self-contained and independent of the machinery of Section~\ref{sec:global}.

\begin{proof}
By the symmetry that exchanges arms and reverses the category order (which maps
$(\theta,\sigma_1^2,\sigma_2^2,m)$ to $(\theta,\sigma_2^2,\sigma_1^2,$ reversed
$m)$ and fixes $\prem(\tfrac12)$), assume $p=e_j$. Then $\sigma_1^2=0$ and
$\prem=\sigma_2^2/\{2v(m)\}$ with $m=(e_j+q)/2$.

\emph{Reduction to three atoms.} Values of $b$ depend only on position relative
to $j$: $b=0$ below, $\tfrac12$ at $j$, $1$ above. Redistributing the mass of
$q$ among categories strictly below $j$ leaves $\theta$ and $\sigma_2^2$
unchanged, while concentrating it on a single category maximizes $\sum_k q_k^3$
over that block (power-sum inequality), hence maximizes $\sum_k m_k^3$ and
minimizes $v(m)$, weakly increasing $\prem$; likewise above $j$. So it suffices
to consider $q$ with masses $(q_-,q_0,q_+)$ on one category below $j$, on $j$,
and on one category above $j$ ($q_-$ or $q_+$ may be $0$; if $j$ is an endpoint
the corresponding block is empty, a special case of what follows).

\emph{The two-variable problem.} With $\theta=q_++\tfrac{q_0}{2}$ and
$s=\theta(1-\theta)$, direct computation gives
\[
\sigma_2^2=s-\frac{q_0}{4},
\qquad
v(m)=\frac1{12}\Bigl[\,1-\Bigl(\tfrac{1+q_0}{2}\Bigr)^{3}-B^3-C^3\Bigr],
\]
where $B=\tfrac{1-\theta-q_0/2}{2}$ and $C=\tfrac{\theta-q_0/2}{2}$ satisfy
$B,C\ge0$, $B+C=\tfrac{1-q_0}{2}=:u$, and
$s-\tfrac{q_0}{4}=4BC+\tfrac{q_0(1-q_0)}{4}$. Writing $w=BC\in[0,u^2/4]$ and
using $B^3+C^3=u^3-3uw$, the claim $\prem\le\tfrac{16}{9}$, i.e.\
$6(s-\tfrac{q_0}{4})\le\tfrac{16}{9}\cdot 12\,v(m)$, rearranges to
\[
1-\Bigl(\tfrac{1+q_0}{2}\Bigr)^{3}-u^3+3uw-\tfrac{27}{32}q_0(1-q_0)
\;\ge\;\tfrac{27}{2}\,w .
\]
Since $u\le\tfrac12$, the coefficient $\tfrac{27}{2}-3u\ge12>0$, so the
inequality is hardest at $w=u^2/4$ (i.e.\ $B=C$). Substituting $w=u^2/4$ and
$u=\tfrac{1-q_0}{2}$ and clearing denominators, the left-minus-right side equals
\[
\frac{3\,q_0\,(1-q_0)\,(q_0+6)}{32}\;\ge\;0 ,
\]
an elementary polynomial identity, with equality iff $q_0\in\{0,1\}$. The case
$q_0=1$ is the null ($q=p$, degenerate); the case $q_0=0$ reduces to the family
of Theorem~\ref{thm:main}(i), whose maximum over $t$ of $16t(1-t)/\{2+t(1-t)\}$ is $16/9$,
attained uniquely at $t=\tfrac12$, with $B=C$ forcing
$q_-=q_+=\tfrac12$.
\end{proof}

\section{Code and data}\label{app:si:code}

The following accompany the manuscript. All results in the main text and in this
document are regenerated by these scripts; random seeds are fixed in each.

\begin{center}
\small
\begin{tabular}{ll}
\toprule
file & contents\\
\midrule
\texttt{core.py} & definitions of $\theta$, $\sigma_1^2$, $\sigma_2^2$, $v(\cdot)$, $\prem$; unit tests\\
\texttt{handproof.py} & symbolic verification of the identities in the proof of Theorem~2\\
\texttt{verify\_reduction.py} & exact rational enumeration of $\Psi\le0$ over full simplices\\
\texttt{convexity\_check.py} & structure of $\Psi$ on $\theta$-slices; extreme-point diagnostics\\
\texttt{boundary.py} & the boundary factorization $\Psi=-q_0G(u,w)$\\
\texttt{zerosets.py} & zero sets of $\Psi$ on the six orbit representatives\\
\texttt{search.py}, \texttt{ccurve.py} & global search and constrained variants (Section~\ref{app:si:search})\\
\texttt{empirical\_real.py} & extracted trial data and scenario panel (Sections~\ref{app:si:panel})\\
\texttt{mechanisms.py} & non-shift mechanisms (Section~\ref{app:si:nonshift})\\
\texttt{finite\_sample.py}, \texttt{calib.py} & power calibration (Section~\ref{app:si:power})\\
\texttt{fig1.py} & Figure~1 of the main text\\
\texttt{empirical\_real\_results.csv} & scenario-level output, $8{,}953$ rows\\
\texttt{mechanisms\_results.csv} & non-shift mechanism output\\
\texttt{real\_datasets.csv} & the five extracted trial comparisons\\
\bottomrule
\end{tabular}
\end{center}

\section*{Declaration of generative AI use}

The author used a large language model (Claude, Anthropic) extensively in preparing this manuscript. Its use was not confined to language editing, and is described here in full.

The model was used interactively to develop the mathematical argument, including the variance-free reformulation of the envelope inequality as the polynomial condition $\Psi\le0$, the convexity reduction of Lemma~\ref{lem:reduce}, the symmetry reduction of Lemma~\ref{lem:sym}, and the algebraic decompositions \eqref{eq:Fdecomp} and \eqref{eq:Sform} underlying the six remaining cases. It wrote the verification code (symbolic checking of the identities, exact rational enumeration, the numerical search, and the scenario panel), performed the extraction and analysis of the trial data reported in Section~\ref{sec:empirical}, conducted literature searches, and drafted the text of the manuscript.

An earlier version of the proof of Theorem~\ref{thm:global} closed six cases by machine-generated Bernstein certificates; that route was subsequently replaced, also with model assistance, by the elementary argument given here. An earlier statement of Theorem~\ref{thm:equality} asserted a global equality characterization that its proof did not support; this was identified during review and corrected to the present, weaker statement.

The author directed the research, set the questions, made all decisions about scope, framing and claims, and takes full responsibility for the content of the manuscript, including the correctness of all mathematical statements and the accuracy of all reported data. All theorems, proofs, numerical results and data extractions have been checked by the author independently of the model. The model is not an author and has not been listed as one.

\section*{Data availability statement}

The data that support the findings of this study are openly available. The IST-3 arm-level counts are published in Table~2 of the primary trial report.\cite{ist3} The shoulder-tip pain data are distributed in the \texttt{nparLD} R package\cite{nparld} and the rheumatoid arthritis data in the \texttt{multgee} R package.\cite{multgee} All code used to derive the results, together with the derived data tables and scenario-level output, is archived at [repository URL / DOI]. This includes the symbolic verification of the algebraic identities in the proof of Theorem~\ref{thm:global}, the exact-arithmetic enumeration underlying Appendix~\ref{sec:verify}, the search and scenario-generation code, and the scripts that regenerate every figure and table.

\section*{Conflict of interest}

The author declares no conflict of interest.

\section*{ORCID}

\noindent Akarin Phaibulpanich \texttt{https://orcid.org/[ORCID iD]}


\begin{thebibliography}{99}\small

\bibitem{brunner2000} Brunner E, Munzel U. The nonparametric Behrens--Fisher problem: asymptotic theory and a small-sample approximation. \emph{Biom J}. 2000;42(1):17-25.

\bibitem{neubert2007} Neubert K, Brunner E. A studentized permutation test for the non-parametric Behrens--Fisher problem. \emph{Comput Stat Data Anal}. 2007;51(10):5192-5204.

\bibitem{happ2019} Happ M, Bathke AC, Brunner E. Optimal sample size planning for the Wilcoxon--Mann--Whitney test. \emph{Stat Med}. 2019;38(3):363-375.

\bibitem{birnbaum1957} Birnbaum ZW, Klose OM. Bounds for the variance of the Mann--Whitney statistic. \emph{Ann Math Stat}. 1957;28(4):933-945.

\bibitem{noether1987} Noether GE. Sample size determination for some common nonparametric tests. \emph{J Am Stat Assoc}. 1987;82(398):645-647.

\bibitem{zhao2008} Zhao YD, Rahardja D, Qu Y. Sample size calculation for the Wilcoxon--Mann--Whitney test adjusting for ties. \emph{Stat Med}. 2008;27(3):462-468.

\bibitem{kieser2003} Kieser M, Friede T. Simple procedures for blinded sample size adjustment that do not affect the type I error rate. \emph{Stat Med}. 2003;22(23):3571-3581.

\bibitem{rustagi1961} Rustagi JS. Bounds for the variance of Mann--Whitney statistic. \emph{Ann Inst Stat Math}. 1961;13:119-126.

\bibitem{bamber1975} Bamber D. The area above the ordinal dominance graph and the area below the receiver operating characteristic graph. \emph{J Math Psychol}. 1975;12(4):387-415.

\bibitem{brunner2025} Brunner E, Konietschke F. An unbiased rank-based estimator of the Mann--Whitney variance including the case of ties. \emph{Stat Pap}. 2025;66:20.

\bibitem{schuurhuis2025} Sch\"uurhuis S, Konietschke F, Brunner E. A new approach to the nonparametric Behrens--Fisher problem with compatible confidence intervals. \emph{Biom J}. 2025;67(1):e70096.

\bibitem{poehlmann2024} P\"ohlmann A, Brunner E, Konietschke F. Sample size planning for rank-based multiple contrast tests. \emph{Biom J}. 2024;66(3):e2300240.

\bibitem{whitehead1993} Whitehead J. Sample size calculations for ordered categorical data. \emph{Stat Med}. 1993;12(24):2257-2271.

\bibitem{ist3} IST-3 collaborative group. The benefits and harms of intravenous thrombolysis with recombinant tissue plasminogen activator within 6 h of acute ischaemic stroke (the third international stroke trial [IST-3]): a randomised controlled trial. \emph{Lancet}. 2012;379(9834):2352-2363.

\bibitem{hoeffding1948} Hoeffding W. A class of statistics with asymptotically normal distribution. \emph{Ann Math Stat}. 1948;19(3):293-325.

\bibitem{bombardier1986} Bombardier C, Ware J, Russell IJ, Larson M, Chalmers A, Read JL. Auranofin therapy and quality of life in patients with rheumatoid arthritis: results of a multicenter trial. \emph{Am J Med}. 1986;81(4):565-578.

\bibitem{lumley1996} Lumley T. Generalized estimating equations for ordinal data: a note on working correlation structures. \emph{Biometrics}. 1996;52(1):354-361.

\bibitem{nparld} Noguchi K, Gel YR, Brunner E, Konietschke F. nparLD: an R software package for the nonparametric analysis of longitudinal data in factorial experiments. \emph{J Stat Softw}. 2012;50(12):1-23.

\bibitem{multgee} Touloumis A. R package multgee: a generalized estimating equations solver for multinomial responses. \emph{J Stat Softw}. 2015;64(8):1-14.

\bibitem{touloumis2013} Touloumis A, Agresti A, Kateri M. GEE for multinomial responses using a local odds ratios parameterization. \emph{Biometrics}. 2013;69(3):633-640.

\bibitem{delacre2017} Delacre M, Lakens D, Leys C. Why psychologists should by default use Welch's t-test instead of Student's t-test. \emph{Int Rev Soc Psychol}. 2017;30(1):92-101.

\bibitem{zimmerman2004} Zimmerman DW. A note on preliminary tests of equality of variances. \emph{Br J Math Stat Psychol}. 2004;57(1):173-181.

\end{thebibliography}
\end{document}